\documentclass[%
 reprint,
superscriptaddress,
 amsmath,amssymb,
 aps,
 prx,
]{revtex4-2}

\usepackage{braket}
\usepackage{graphicx}
\usepackage{dcolumn}
\usepackage{bm}
\usepackage{amsthm}

\graphicspath{{./figures/}}

\usepackage{tikz}
\usepackage{float}
\makeatletter
\let\newfloat\newfloat@ltx
\makeatother
\usepackage[linesnumbered, ruled, vlined]{algorithm2e}
\usepackage[colorlinks=true, hyperindex, breaklinks, linkcolor=blue, urlcolor=blue, citecolor=blue]{hyperref} 
\usepackage{float}
\usepackage{eqnarray,amsmath}
\usepackage{tikz}
\usepackage{appendix}
\usepackage{makecell}

\newtheorem{definition}{Definition}
\newtheorem{proposition}{Proposition}
\usepackage{bbm}

\newcommand{\mc}{\mathcal}
\newcommand{\mb}{\mathbf}

\newcommand{\mr}{\mathrm}
\newcommand{\kket}[1]{|#1\rangle\rangle}
\newcommand{\bbra}[1]{\langle\langle #1|}

\begin{document}

\preprint{APS/123-QED}

\title{Fault-Tolerant Operation of Bosonic Qubits with Discrete-Variable Ancillae} 
\author{Qian Xu}
\affiliation{Pritzker School of Molecular Engineering, The University of Chicago, Chicago 60637, USA}

\author{Pei Zeng}
\affiliation{Pritzker School of Molecular Engineering, The University of Chicago, Chicago 60637, USA}

\author{Daohong Xu}
\affiliation{Pritzker School of Molecular Engineering, The University of Chicago, Chicago 60637, USA}

\author{Liang Jiang}
\email{liang.jiang@uchicago.edu}
\affiliation{Pritzker School of Molecular Engineering, The University of Chicago, Chicago 60637, USA}

\date{\today}

\begin{abstract}

Fault-tolerant quantum computation with bosonic qubits often necessitates the use of noisy discrete-variable ancillae. In this work, we establish a comprehensive and practical fault-tolerance framework for such a hybrid system and synthesize it with fault-tolerant protocols by combining bosonic quantum error correction (QEC) and advanced quantum control techniques. 
We introduce essential building blocks of error-corrected gadgets by leveraging ancilla-assisted bosonic operations using a generalized variant of path-independent quantum control (GPI).
Using these building blocks, we construct a universal set of error-corrected gadgets that tolerate a single photon loss and an arbitrary ancilla fault for four-legged cat qubits. Notably, our construction only requires dispersive coupling between bosonic modes and ancillae, as well as beam-splitter coupling between bosonic modes, both of which have been experimentally demonstrated with strong strengths and high accuracy. Moreover, each error-corrected bosonic qubit is only comprised of a single bosonic mode and a three-level ancilla, featuring the hardware efficiency of bosonic QEC in the full fault-tolerant setting. 
We numerically demonstrate the feasibility of our schemes using current experimental parameters in the circuit-QED platform.
Finally, we present a hardware-efficient architecture for fault-tolerant quantum computing by concatenating the four-legged cat qubits with an outer qubit code utilizing only beam-splitter couplings.
Our estimates suggest that the overall noise threshold can be reached using existing hardware. 
These developed fault-tolerant schemes extend beyond their applicability to four-legged cat qubits and can be adapted for other rotation-symmetrical codes, offering a promising avenue toward scalable and robust quantum computation with bosonic qubits.

\end{abstract}

\maketitle
\section{Introduction}

Quantum error correction (QEC) enables reliable quantum information processing~\cite{nielsen_chuang_2010, lidar_brun_2013, aharonov1996limitations}. However, paradigmatic QEC schemes, particularly those employing surface codes with physical qubits~\cite{kitaev2003fault, bravyi1998quantum, fowler2012surface, google2023suppressing}, suffer from huge resource overhead~\cite{litinski2019game, beverland2022assessing}.
This resource-intensive nature creates a substantial gap between the theoretical potential of fault tolerance and the capabilities of current noisy intermediate-scale quantum (NISQ)~\cite{Preskill2018quantumcomputingin} devices. 

Encoding quantum information into bosonic systems~\cite{gottesman2001encoding, chuang1997bosonic, michael2016new, cochrane1999macroscopically, albert2018performance} by leveraging their infinite-dimensional Hilbert spaces offers a promising avenue to reduce the overhead of QEC~\cite{darmawan2021practical, chamberland2022building, xu2023autonomous, guillaud2019repetition, gouzien2023performance, noh2022low}. 
While robust quantum memories based on single-mode bosonic codes have been experimentally demonstrated with improved memory lifetime~\cite{ofek2016extending, sivak2023real, ni2023beating}, realizing error-corrected operations on these bosonic qubits remains a formidable task.




One of the primary complexities stems from the weak non-linear interactions inherent in bosonic modes, necessitating the use of discrete-variable ancillae in systems such as circuit quantum electrodynamics (circuit QED) platform~\cite{blais2004cavity, krastanov2015universal}. 
However, a significant challenge arises in this hybrid system, as errors in the ancillae tend to propagate back to the bosonic mode, potentially compromising the encoded quantum information~\cite{heeres2017implementing}.
To address this issue, several methods have been developed to maintain precise control over the bosonic mode even in the presence of noisy ancillary systems~\cite{ma2020error, reinhold2020error, ma2020path}. Nevertheless, a comprehensive fault-tolerance framework for this hybrid system, along with guidelines for constructing fully fault-tolerant protocols using advanced quantum control concepts, remains conspicuously absent. Consequently, while universal error-detection operations on bosonic qubits have been constructed~\cite{teoh2022dual, tsunoda2023error} and demonstrated~\cite{chou2023demonstrating}, achieving a complete set of error-corrected operations has remained a significant challenge.

In this work, we bridge this gap by introducing a fault-tolerance framework tailored to the hybrid system composed of bosonic data modes and discrete-variable ancillae. 
Inspired by concatenated qubit codes~\cite{aliferis2005quantum}, we identify essential properties for gadgets encoded in bosonic codes (referred to as ``level-1" gadgets) in Sec.~\ref{sec:fault_tolerance}. These properties play a crucial role in determining the fault tolerance of a level-1 circuit, where the overall failure probability must be suppressed to a certain order of the physical error rate. Furthermore, we demonstrate how the defined fault tolerance can be achieved through the integration of bosonic QEC with compatible quantum control techniques. Specifically, in Secs.~\ref{Sec:GeneralizedPI} and~\ref{Sec:GPItoFT}, we establish a connection between a generalized version of path-independent control~\cite{ma2020path} (referred to as GPI) and fault tolerance, highlighting the importance of GPI operations as fundamental building blocks for error-corrected gadgets.

As an application of these fault-tolerant tools, in Sec.~\ref{sec:FT4legCat}, we construct universal error-corrected gadgets using GPI operations for the four-legged cat qubit~\cite{cochrane1999macroscopically, leghtas2013hardware, albert_performance_2018}. 
These gadgets can tolerate a single photon loss and an arbitrary ancilla fault, while only relying on dispersive coupling between bosonic modes and ancillae~\cite{heeres2015cavity, rosenblum2018fault, reinhold2020error} and beam-splitter (BS) coupling between bosonic modes~\cite{chapman2023high, lu2023high}.  Importantly, these coupling mechanisms have been experimentally demonstrated with strong coupling strengths. 
Each level-1 logical qubit, encoded in a four-legged cat code, utilizes only a single bosonic mode and a three-level ancilla, featuring the hardware efficiency of bosonic QEC. 
We numerically demonstrate the second-order error suppression for the level-$1$ gadgets. 
Moreover, we show that using a teleportation gadget that pumps energy into the system and suppresses phase-rotation errors, a robust cat-encoding memory is feasible even in the presence of finite $\chi$ mismatches in the circuit-QEC platform with current experimental parameters~\cite{reinhold2020error}. 

Finally in Sec.~\ref{Sec:concatenatedQEC}, we present a practical and hardware-efficient architecture for fault-tolerant quantum computing by concatenating the four-legged cat qubits with an outer qubit code.
While we primarily focus on the four-legged cat code throughout this work, we discuss in Sec.~\ref{sec:discussion} that the fault-tolerant schemes developed herein can be readily adapted to other rotation-symmetric bosonic codes~\cite{grimsmo2020quantum}.

\section{System description and error model} \label{Sec:System}

We first introduce some notations. 
We denote $[k] := \{1,2,\cdots,k\}$ as the set of integers from $1$ to $k$.
We denote $\left[ \int_{t_h} dt_h \right]_{h \in [k]} := \int_{t_k} d_{t_k}\int_{t_{k-1}}dt_{k-1}\cdots \int_{t_1}dt_1 $ as the multiple integral over variables in $\{t_h\}_{h\in[k]}$, and similarly $\left[\sum_{a_h}\right]_{h\in[k]} := \sum_{a_k }\sum_{a_{k-1}}\cdots\sum_{a_1}$ as the sum over variables in $\{a_h\}_{h \in [k]}$. We denote $A \propto B$ if there exists some $c \in \mathbb{C}$ such that $A = c B$. We denote $\mc{T}$ as the time ordering operator.

\subsection{Preliminaries}

\subsubsection{Bosonic codes \label{sec:bosonic_code}}
Single-mode bosonic error-correcting codes encode logical information into a subspace of the infinite-dimensional Hilbert space of an oscillator. Among them, the four-legged cat code~\cite{cochrane1999macroscopically, leghtas2013hardware, albert_performance_2018} encodes a single logical qubit and has codewords
\begin{equation}
    \ket{\mu_L} = c_{\mu}\left[ \ket{\alpha} + \ket{-\alpha} + (-1)^{\mu}(\ket{i\alpha} + \ket{-i\alpha}) \right],
\end{equation}
where $\mu = 0/1$, $\ket{\gamma}$ denoites a coherent state with an amplitude $\gamma \in \mathbb{C}$, and $c_{\mu} = 1/(2\sqrt{2\exp(-|\alpha|^2)(\cosh|\alpha|^2 + (-1)^\mu \cos|\alpha|^2)})$ are normalization constants. Given any quantum code encoding a single logical qubit, we denote $P_c := \ket{0_L}\bra{0_L} + \ket{1_L}\bra{1_L}$ as the projection onto the codespace, and $\bar{X}_c, \bar{Y}_c, \bar{Z}_c$ the logical $X$-, $Y$-, $Z$-Pauli operators respectively.

The capability of an error-correcting code to correct a given set of errors $\mb{E}$ is given by the Knill-Laflamme (KL) condition~\cite{knill1997theory}: $P_c E_i^{\dagger} E_j P_c \propto P_c$ for any $E_i, E_j \in \mb{E}$. More specifically, we can evaluate the $2\times 2$ QEC matrix $\epsilon_{jk}^c$ for any pair of errors $E_j, E_k$~\cite{albert_performance_2018}:
\begin{equation}
    P_c E_j^{\dagger} E_k P_c = \epsilon_{jk}^c,
\end{equation}
where $\epsilon_{jk}^c$ can be parametrized as $\epsilon_{jk}^c = c_{jk}^c P_c + x_{jk}^c \bar{X}_c + y_{jk}^c \bar{Y}_c + z_{jk}^c \bar{Z}_c$, where $c_{jk}^c, x_{jk}^c, y_{jk}^c, z_{jk}^c \in \mathbb{C}$. The KL condition is satisfied if $x_{jk}^c = y_{jk}^c = z_{jk}^c = 0$ for any $j$ and $k$.

Consider the four-legged code and an error set containing a single-photon loss $\mb{E} = \{I, a\}$, where $a$ denotes the annihilation operator. First, we have $P_c a P_c = 0$, indicating that a single-photon loss is perfectly detectable. Second, 
\begin{equation}
    P_c a^{\dagger}a P_c = \bar{n} P_c + \frac{\delta n}{2} \bar{Z}_c, 
\end{equation}
where $\bar{n} := (\bra{0_L}a^{\dagger}a\ket{0_L} + \bra{1_L}a^{\dagger}a\ket{1_L})/2$ denotes the mean photon number and $\delta n := \bra{0_L}a^{\dagger}a\ket{0_L} - \bra{1_L}a^{\dagger}a\ket{1_L}$ denotes the photon number difference between the two codewords. For an arbitrary $\alpha$, $\delta n \neq 0$, indicating the a single photon loss is not perfectly correctable. However, $\delta n = O(e^{-2\alpha^2})$ as $\alpha \gg 1$ and a single-photon loss is approximately correctable for large $\alpha$. Furthermore, $\delta n = 0$ is exactly satisfied at a discrete set of finite $\alpha$~\cite{li2017cat}, which we refer to as sweet spots.  Similarly, one can show that for a continuous set of phase-rotation errors $\mb{R} = \{e^{i\theta a^{\dagger}a}\}_{\theta \in [-\theta_m,\theta_m]}$, the KL condition is approximately satisfied for large $\alpha$ if $\theta_m < \pi/4$~\cite{grimsmo2020quantum}. First, $P_c e^{-i\theta_1 a^{\dagger}a}e^{i\theta_2 a^{\dagger}a}P_c = c_{12} P_c + z_{12}\bar{Z}_c$ for any $\theta_1, \theta_2 \in \mb{R}$ since $e^{i(\theta_2 - \theta_1)a^{\dagger}a}$ preserves the photon number. Next, 
\begin{equation}
\begin{aligned}
    z_{12} & = \left(\bra{+_L} e^{i(\theta_2 - \theta_1)a^{\dagger}a}\ket{-_L} + \bra{-_L} e^{i(\theta_2 - \theta_1)a^{\dagger}a}\ket{+_L}\right)/2 \\
    & \approx \left(\langle i\alpha|\alpha e^{i(\theta_2 - \theta_1)}\rangle + \langle -i\alpha|\alpha e^{i(\theta_2 - \theta_1)}\rangle\right)/2 + h.c.,\\
\end{aligned}
\end{equation}
where the the approximation utilizes that $\ket{+_L} \approx (\ket{\alpha} + \ket{-\alpha})/\sqrt{2}$ and $\ket{-_L} \approx (\ket{i\alpha} + \ket{-i\alpha})/\sqrt{2}$ for large $\alpha$. Obviously, $z_{12} \rightarrow 0$ as $\alpha \gg 1$ as long as $|\theta_2 - \theta_1| \neq \pi/2$, which holds if $\theta_m < \pi/4$. 

To conclude, the four-legged cat code can approximately correct a single photon loss and a continuous set of phase rotations with amplitude smaller than $\pi/4$ (for large $\alpha$). 
In fact, cat codes serve as numerically optimized codes for certain regimes of a bosonic channel with both loss and dephasing errors~\cite{leviant2022quantum}. 


\subsubsection{Open quantum system and Markovian quantum evolution}
A noisy Markovian evolution of a quantum system is described by a Lindblad master equation:
\begin{equation}
\begin{aligned}
    \frac{d \rho}{d t} & = \mc{L}(t) \rho =
    - i [H(t), \rho] + (\sum_j  \mathcal{D}[\sqrt{\gamma_j} J_j])\rho,
    \label{eq:Lindblad_equation}
\end{aligned}
\end{equation}
where $H(t)$ is the system Hamiltonian and $\mathcal{D}[O] = O\bullet O^{\dagger} - \frac{1}{2}\{O^{\dagger}O, \bullet\}$ is a Lindblad superoperator associated with a jump operator $O$, and $\gamma_j$ is the jump rate for the error $J_j$. 

Denote $H_{\mathrm{eff}}(t) := H(t) - \frac{i}{2}\sum_j \gamma_j J_j^{\dagger}J_j$ as the effective Hamiltonian, and $\mathcal{S} := \sum_j \gamma_j J_j \bullet J_j^{\dagger}$ as the superoperator describing quantum jumps. The Lindbladian $\mc{L}(t)$ can be rewritten as $\mc{L}(t) = -i[H_{\mr{eff}}(t), \bullet] + \mc{S}$. Then, the joint quantum channel, given by the time integral of Eq.~\eqref{eq:Lindblad_equation}, admits a Dyson expansion with respect to $\mc{S}$~\cite{ma2020path}:
\begin{equation}
    \begin{aligned}
        \rho(t) = \mathcal{G}(t, 0)\rho(0) = \sum_{q = 0}^{\infty} \mathcal{G}_q(t, 0) \rho(0),
    \end{aligned}
\end{equation}
where $\mathcal{G}_0(t,0) = \mathcal{W}(t, 0) := W(t,0)\bullet W^{\dagger}(t,0)$, with $W(t, 0) := \mathcal{T} \exp \left[-i \int_{0}^{t} H_{\mathrm{eff}}\left(t^{\prime}\right) d t^{\prime}\right]$, describes evolution without any quantum jumps, and
\begin{equation}
    \begin{aligned}
\mathcal{G}_q(t)= \left[\int_{t_h = 0}^t dt_h\right]_{h \in [q]} \mc{T} \big( & \mathcal{W}(t, t_q)\mathcal{S} \cdots \\ & \mathcal{S} \mathcal{W}(t_2, t_1) \mathcal{S W}(t_1, 0) \big),
\end{aligned}
\label{eq:dyson_expansion}
\end{equation}
where $\mathcal{G}_q$ ($q \geq 1$) describes the evoluation with $q$ quantum jumps. We refer to such an expansion as the jump expansion and and $\mc{G}^{[n]} := \sum_{q = 0}^n \mc{G}_q$ as the $n$-th order truncation of $\mc{G}$ under the jump expansion.

For quantum error correction, we care about the expansion of the channel $\mc{G}$ in terms of the small noise parameter $p := \gamma_i t$ given an evolution time $t$ (here, we have assumed equal noise rates for simplicity):
\begin{equation}
\mc{G}(t,0) = \sum_{q^{\prime}}p^{q^{\prime}}\mathcal{G}^{\prime}_{q^{\prime}}(t,0).
\end{equation}
Such an expansion can be obtained by Dyson expanding $\mc{G}$ with respect to the full Lindblad superoperators of the noise $(\sum_j \mathcal{D}[\sqrt{\gamma_j} J_j])$ in Eq.~\eqref{eq:Lindblad_equation}, instead of their quantum-jump components $\mc{S}$.  
We refer to such an expansion of $\mc{G}$ as its noise expansion, and $\mc{G}^{\prime [n]} := \sum_{q^{\prime} = 0}^{n} \mc{G}^{\prime}_q$ as the $n$-th order truncation of $\mc{G}$ under the noise expansion. 

Observe that $\mc{G}^{[n]} = \mc{G}^{\prime [n]} + O(p^{n + 1})$, i.e. the $n$-th order truncation of a channel $\mc{G}$ in terms of its jump expansion or its noise expansion is equivalent up to $n$-order of $p$. Since $\mc{G}^{[n]}$ is easier to evaluate for the purpose of this work, we will mostly consider the jump expansion of channels. A nice property of a channel's jump expansion is that it is automatically in a Kraus form:
\begin{equation}
\begin{aligned}
     \mc{G}(t,0) = \sum_{q = 0} & \left[\int_{t_h = 0}^t dt_h\right]_{h \in [q]} \left[\sum_{j_h=1}^N\right]_{h \in [q]} \\
     & G_q(\{t_h\}, \{j_h\})\bullet G_q^{\dagger}(\{t_h\}, \{j_h\}),
\end{aligned}
\label{eq:dyson_expansion_Kraus_form}
\end{equation}
where 
\begin{equation}
    G_q(\{t_h\}, \{j_h\}):= \mc{T}  W(T,t_q)E_{j_q} \cdots E_{j_2} W(t_2, t_1) E_{j_1} W(t,0).
\end{equation}
One can, therefore, view $G_q(\{t_h\}, \{j_h\})$ as a Kraus operator of the channel with discrete indices $q, \{j_h\}$ and continuous indices $\{t_h\}$.

\subsection{General setup} \label{ssec:setup}
\begin{figure}
    \centering
    \includegraphics[width=0.48\textwidth]{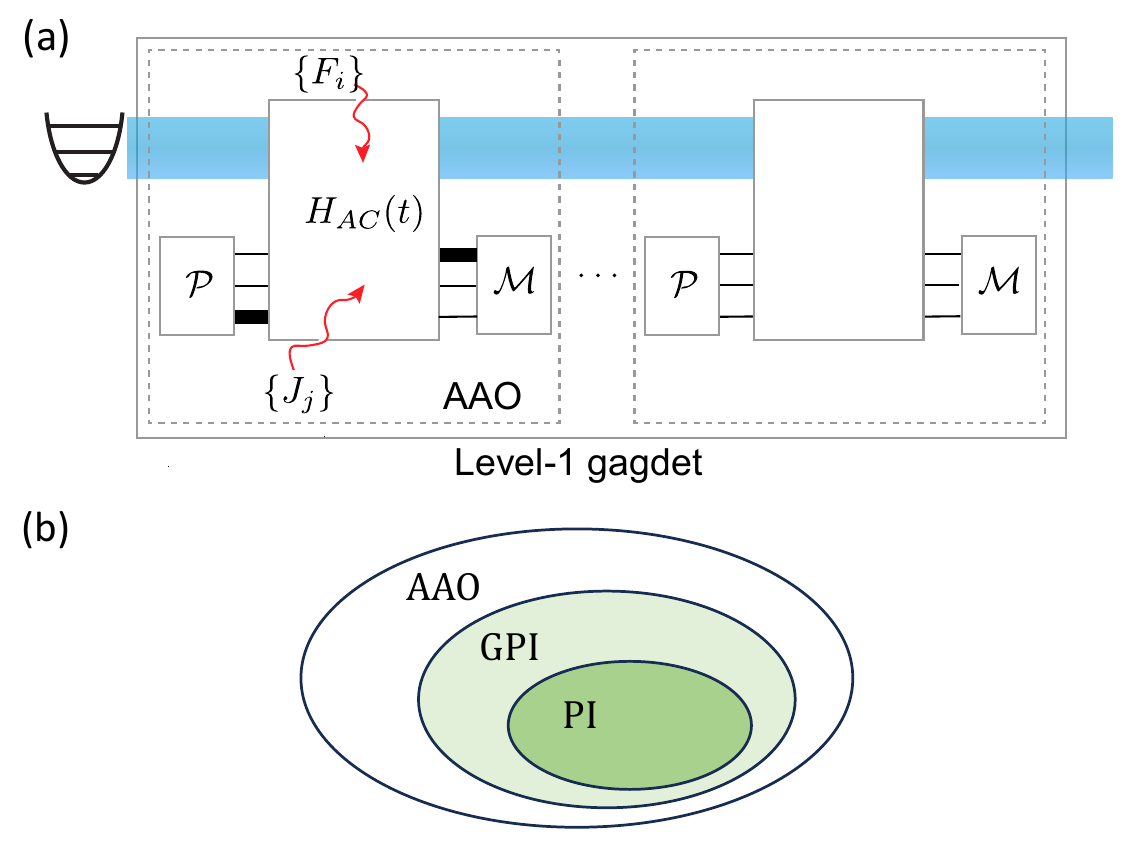}
    \caption{(a) Illustration of a level-$1$ bosonic gadget consisting of a sequence of ancilla-assisted operations. For each AAO, the ancilla is initialized to some state $\ket{i}$ and measured in some basis $\mc{B}_A$. The later AAOs can be performed adaptively using the earlier ancilla measurement outcomes. (b) Illustration of the AAO, GPI and PI operations. As a special case of AAO, the GPI operations with bosonic QEC can handle bosonic errors induced by relevant ancilla faults. The previous PI operations~\cite{ma2020path} can be regarded as a special GPI without bosonic QEC, which are designed to avoid any introduction of bosonic errors due to relevant ancilla faults.}
    \label{fig:gadget}
\end{figure}
As shown in Fig.~\ref{fig:gadget}(a), we consider gadgets for some bosonic code consisting of a sequence of ancilla-assisted operations (AAOs). For each AAO, a $d_A \geq 2$ ancilla $A$ is initialized in some initial state $\ket{i}_A$, interacts with the bosonic mode $C$ for some time $T$, and is finally measured in some basis $\mathcal{B}_A$. 
We consider continuous Markovian interactions between $A$ and $C$, which is described by the Lindblad master equation in Eq.~\eqref{eq:Lindblad_equation} with a Hamiltonian $H_{AC}(t)$ that acts on the joint system, a set of bosonic jump operators $\{\sqrt{\kappa_i} F_i \}$, and a set of ancilla jump operators $\{\sqrt{\gamma_j} J_j \}$. We allow adaptively choosing the later AAOs using the earlier ancilla measurement outcomes. Note that a direct operation on the bosonic mode can be viewed as a trivial AAO with the ancilla being initialized in some state $\ket{i}$, idling for the evolution time, and measured in $\ket{i}$  without any errors. 

Such an AAOs-composed bosonic gadget forms a channel $\mc{N}$ on the bosonic mode, which can be decomposed as $\mathcal{N} = \mathcal{N}_{n} \circ \mathcal{N}_0$, where $\mathcal{N}_0$ is the target bosonic operation and $\mathcal{N}_n = \sum_k N_k \bullet N_k^{\dagger}$ is a data noise channel represented by a set of noise Kraus operators $\{N_k\}$.  Fault-tolerance essentially concerns the relation between the physical channels $\{\mathcal{G}\}$ and the resultant bosonic channel $\mathcal{N}$. More specifically, we need to study how the noise in $\mathcal{G}$, which we refer to as faults, propagates to the data errors $\{N_k\}$ in $\mathcal{N}_n$. We will need to quantify the size of the faults and the data errors and design AAOs such that small faults only propagate to small data errors. 
We refer to a physical channel $\mc{G}$ that contains no more than $t$ faults as its $n$-th order truncation $\mc{G}^{[n]}$.  
To quantify the size of the data bosonic errors, we need to specify a bosonic error-correcting code and an error basis.  In this work, we will primarily focus on the cat codes~\cite{cochrane1999macroscopically, leghtas2013hardware, mirrahimi2014dynamically} and a basis we termed loss-dephasing basis, which is closely related to photon loss and bosonic dephasing errors. 

\subsection{Loss-dephasing basis and error metrics} 
Typical bosonic errors include excitation loss ($a$), heating ($a^{\dagger}$), and bosonic dephasing ($a^{\dagger}a$). For such errors, a natural basis to use is $\{e^{-i\theta a^{\dagger}a} a^k, a^{\dagger k^{\prime}} e^{i\theta^{\prime} a^{\dagger}a}\}_{k, k^{\prime} \in \mathbb{N}; \theta, \theta^{\prime} \in (-\pi, \pi]}$, which is a complete basis for all single-mode bosonic operators.
Neglecting heating errors $a^{\dagger}$, which are typically small~\cite{reinhold2020error, teoh2022dual}, the relevant errors are then spanned by $\{E_{k}(\theta) := e^{-i\theta a^{\dagger}a} a^k\}_{k \in \mc{N}, \theta \in (-\pi,\pi]}$, which we refer to as the loss-dephasing basis. 


A four-legged cat code can only correct errors $E_k(\theta)$ with small $k$ and $|\theta|$ (see Sec.~\ref{sec:bosonic_code}). This motivates us to quantify the size of $E_k(\theta)$ as $|E_k(\theta)|_w := (k, |\theta|) \in \mathbb{N}\times [0, \pi]$. We compare the size of two errors by introducing a partial order with respect to the proper cone $R^{2}_+ := [0,\infty)\times[0,\infty)$, i.e. $|E_{k^{\prime}}(\theta^{\prime})|_w \geq |E_{k}(\theta)|_w \leftrightarrow (k^{\prime} - k, |\theta^{\prime}| - |\theta|) \in R^{2}_+$.  We say that a bosonic noise channel $\mathcal{N}_n$ contains at most $(k, \theta)$ errors if all its Kraus operators have size at most $(k, \theta)$, and a state $\ket{\phi^{\prime}}$ is at most $(k, \theta)$ far from a target state $\ket{\phi}$ if there exists a $\mathcal{N}_n$ containing at most $(k, \theta)$ errors such that $\ket{\phi^{\prime}}$ is in the support of $\mathcal{N}_n (\ket{\phi}\bra{\phi})$.  With this quantification of error sizes, for $\alpha \gg 1$, the four-legged cat can correct errors $|E_k(\theta)|_w \leq (1, \pi/4)$~\cite{albert_performance_2018}. Fig.~\ref{fig:FTfm}(a) illustrates the $2$-dimensional error space indicated by the number of photon loss and dephasing angle.

\begin{figure}
    \centering
    \includegraphics[width=0.48\textwidth]{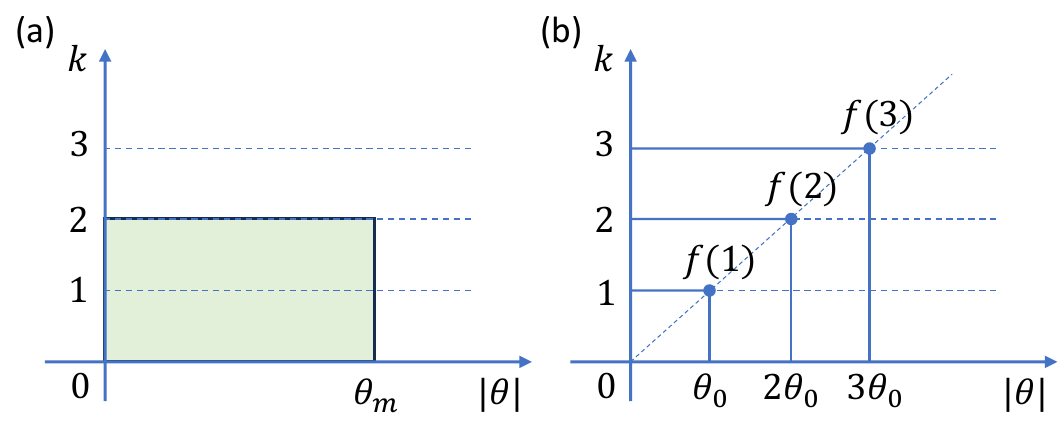}
    \caption{Illustration of bosonic error decomposition and the error propagation function $f(m)$. (a) The bosonic loss-dephasing error can be expanded by the basis $E_k(\theta)$. By defining a partial order of the size $E_k(\theta)$, the bosonic error $E_k(\theta)$ with at most $(k,\theta_m)$ error can be illustrated as the green region in the plot. Here $k=2$. (b) Suppose $m$ faults occur during the gate implementation. To capture the propagation of faults to the final bosonic error, we introduce a function $f(m)=(m,m\theta_0 \text{ mod } \pi)$ as a upper bound of the induced final loss and dephasing errors.
    }
    \label{fig:FTfm}
\end{figure}


\section{fault-tolerance \label{sec:fault_tolerance}}
In this section, we formalize a fault-tolerance framework for the hybrid system with bosonic modes and discrete-variable ancillae in the context of concatenated codes~\cite{aliferis2005quantum}.  
Since the single-mode cat code alone cannot arbitrarily suppress logical errors, one needs to concatenate it with an outer qubit code for fault-tolerant quantum computing.
That is, we will have three levels of gadgets. The level-$0$ gadgets are the physical operations; The level-$1$ gadgets are encoded in the cat code and the level-$2$ gadgets are encoded in the outer qubit code.
A quantum circuit is fault-tolerantly executed using level-$2$ gadgets, and each level-$2$ gadget is executed using a level-$1$ circuit with multiple level-$1$ gadgets. 
In order for each level-$2$ gadget (or equivalent, a level-$1$ circuit) to be executed with a failure rate $O(p^{t + 1})$, which suppresses the physical error rate $p$ to certain order $t + 1$, the level-$1$ gadgets suffice to satisfy the following properties:  

First, there exists a function $f:\mathbb{N} \rightarrow \mathbb{N}\times [0, \pi]$ that satisfies:
\begin{enumerate}
    \item $f(m_1) \leq f(m_2) \leftrightarrow m_1 \leq m_2$ if $m_1, m_2 \leq t$.
    \item $f(m_1 + m_2) = f(m_1) + f(m_2)$ if $m_1 + m_2 \leq t$.
\end{enumerate}
Roughly speaking, $f(m)$ constraints the maximal size of data errors that $a$ faults during a protocol can propagate to. For instance, for a bosonic code that can correct phase rotations smaller than $\theta_{\mr{max}}$, we will choose $f(m) = (m, m\theta_0 \mod \pi)$ for some $\theta_0 \in [0, \theta_{\mr{max}} /t]$, which constraints that $m$ faults can propagate to at most $m$ photon losses and phase rotations of size at most $m \theta_0$. We illustrate such an error propagation constrained by $f$ in Fig.~\ref{fig:FTfm}(b). 


Given $f$ and $t$, we then define $t$-FT fault-tolerant gadgets, including gate, error correction, state preparation, and measurement, for the hybrid system by generalizing the definitions in Ref.~\cite{aliferis2005quantum} for qubits. We remark that, the following FT definitions are related to the choice of the function $f$.


\begin{definition}[$t$-FT gate]
A gate is $t$-FT if it satisfies: For an input codeword that has an error of size $(k, \theta)$, if at most $m$ faults occur during the gate with $(k, \theta) + f(m) \leq f(t)$, the output state is at most $(k, \theta) + f(m)$ far from the codespace; Furthermore, ideally decoding the output state gives the same codeword as first ideally decoding the input state and then applying the ideal gate. 

\label{def:FT_gates}
\end{definition}

Note that this condition for the gate corresponds to the combination of Property 2 and Property 3 of Ref.~\cite{aliferis2005quantum}.

\begin{definition}[$t$-FT QEC] 
A QEC gadget is $t$-FT if it satisfies:
\begin{enumerate}
    \item For an input codeword with an error of size $(k, \theta)$, if at most $m$ faults occur during the protocol with $(k, \theta) + f(m) \leq f(t)$, ideally decoding the output state gives the same codeword as ideally decoding the input state. 
    \item For at most $m$ faults during the protocol with $f(m) \leq f(t)$, no matter the size of the error on the input state, the output state is at most $f(m)$-far from the code space.
\end{enumerate}
\label{def:FT_QEC}
\end{definition}
Note that conditions (i) and (ii) correspond to Properties 1 and 0 of Ref.~\cite{aliferis2005quantum}, respectively.

State preparation and measurement are special cases of FT gates:
\begin{definition}[$t$-FT state preparation]
A state-preparation gadget is $1$-FT if it satisfies:  
If at most $m \leq t$ faults occur during the protocol, the output state is at most $f(m)$-far from the target state; Furthermore, ideally decoding the output state gives the ideal target state. 
\label{def:FT_state_preparation}
\end{definition}

\begin{definition}[$t$-FT measurement]
A measurement gadget is $t$-FT if it satisfies:
For an input codeword that has an error of size $(k, \theta)$, if at most $m$ faults occur during the gate with $(k, \theta) + f(m) \leq f(t)$, the measurement is equivalent to applying the ideal measurement to the input codeword. 
\label{def:FT_measurement}
\end{definition}

Based on the definition of the $t$-FT gadgets, we have the following proposition.

\begin{figure}
    \centering
    \includegraphics[width=0.35\textwidth]{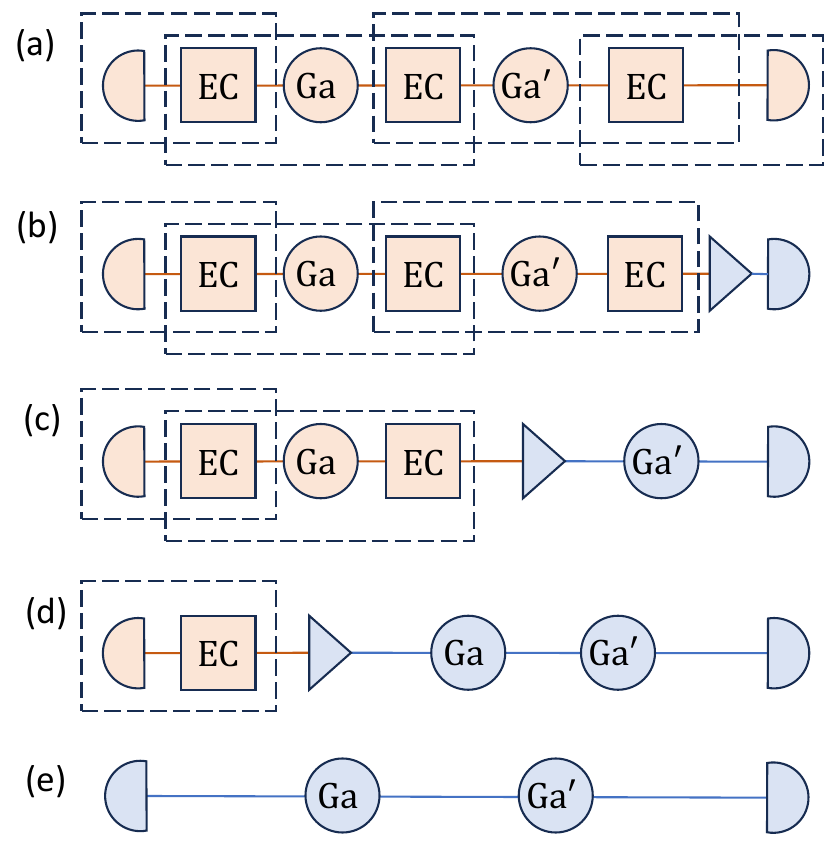}
    \caption{Reduction of a FT level-1 circuit to the ideal circuit.}
    \label{fig:FTreduce}
\end{figure}

\begin{figure}
    \centering
    \includegraphics[width=0.35\textwidth]{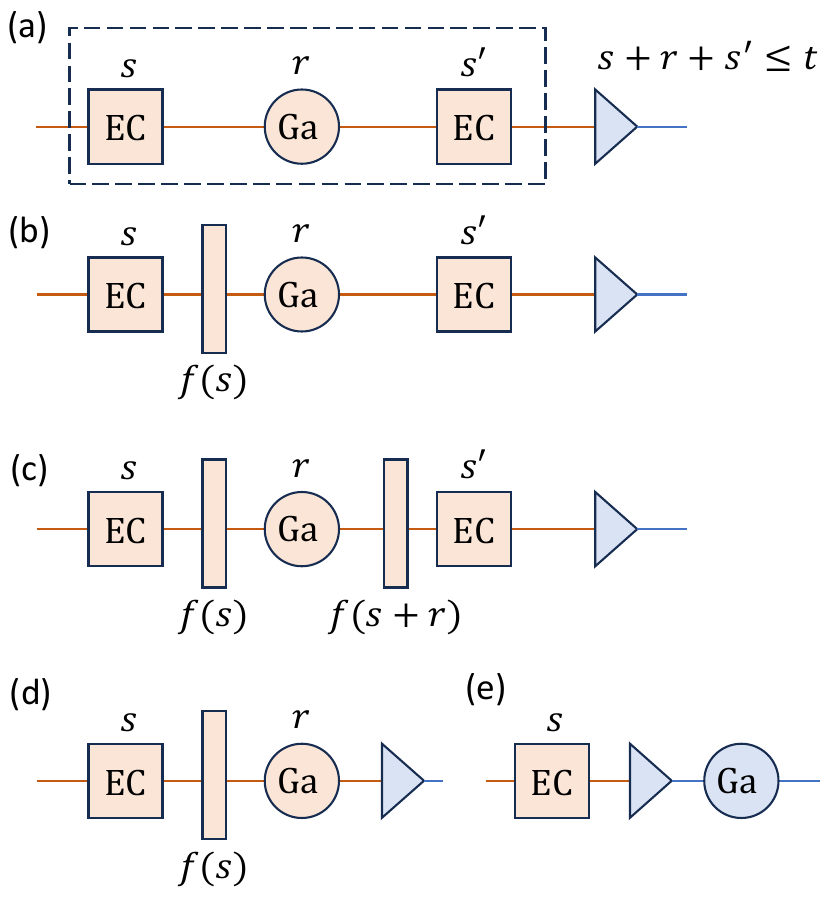}
    \caption{Reduction of the extended rectangular to an ideal gadget.}
    \label{fig:ExRec}
\end{figure}

\begin{proposition}
    Using $t$-FT level-$1$ gadgets, any level-$1$ circuit has a failure rate $O(p^{t + 1})$, where $p\in[0,1)$ is the physical error rate, i.e., the probability that one fault happens in the gadget.
    \label{prop:correct_circuit}
\end{proposition}
\begin{proof}

    We follow the extended-rectangle formalism in Ref.~\cite{aliferis2005quantum}. Without loss of generality, we consider an ideal quantum circuit in Fig.~\ref{fig:FTreduce}(e). 
    Here we take the single-qubit level-1 circuit as an example. In practice, we realize this circuit using the noisy $t$-FT level-1 gadgets shown in Fig.~\ref{fig:FTreduce}(a). To analyze the fault-tolerance property of this circuit, we draw several dashed boxes to cover the whole circuit. The dashed boxes are called extended rectangles $\mathrm{exRec}$. For a quantum gate, an extended rectangle $\mathrm{exRec}$ (a dashed box in Fig.~\ref{fig:FTreduce}(a)) is a composition of a front EC gadget, a gate and a rear EC gadget, i.e. $\mathrm{exRec} = \mathrm{EC}\circ \mathrm{Ga}\circ \mathrm{EC}$.



    We say that any operation $\mathrm{Op}$ is $t$-good if it contains no more than $t$ faults. In what follows, we show that if all the dashed boxes in Fig.~\ref{fig:FTreduce}(a) are $t$-good, we can reduce the noisy circuit to the ideal circuit following the stream in Fig.~\ref{fig:FTreduce}.
    To this end, we introduce the ideal decoder $\mathrm{ID}$ (the blue triangles in Fig.~\ref{fig:FTreduce} and \ref{fig:ExRec}), which performs an ideal recovery given a bosonic code. We also introduce a $(k,\theta)$-filter $[(k,\theta)]\mathrm{F}$ (the orange thin rectangles in Fig.~\ref{fig:ExRec}) which performs a projection onto the space spanned by all states that can be obtained by acting on a codeword with an error no larger than $(k,\theta)$.

    First of all, we notice that if the last box in Fig.~\ref{fig:FTreduce}(a) is $t$-good, then based on the definition of $t$-FT QEC and measurement, we can equivalently convert the circuit in Fig.~\ref{fig:FTreduce}(a) to (b). 
    
    Then, we follow the procedures in Fig.~\ref{fig:ExRec} to reduce the extended gadgets of quantum gates to the ideal gadgets: 
    Denote the faults occur in the front EC gadget, the gate gadget and the rear EC gadget to be $s, r, s'$, respectively. Since the dashed box is $t$-good, we have $s+r+s'\leq t$. Fig.~\ref{fig:ExRec}(a) and (b) are equivalent due to the second requirement of FT QEC in Def.~\ref{def:FT_QEC}; (b) and (c) are equivalent due to the first requirement of the FT gate in Def.~\ref{def:FT_gates}; (c) and (d) are equivalent due to the first requirement of FT QEC in Def.~\ref{def:FT_QEC}; (d) and (e) are equivalent due to the second requirement of the FT gate in Def.~\ref{def:FT_gates}. Then we can transform the circuit from Fig.~\ref{fig:FTreduce}(b) to (d) using the reduction in Fig.~\ref{fig:ExRec}. 

    Finally, we use the property of FT state preparation to reduce Fig.~\ref{fig:FTreduce}(d) to (e). The argument is similar to the ones for the extended gadgets of quantum gates in Fig.~\ref{fig:ExRec}.

    In our gadget set-up in Sec.~\ref{ssec:setup}, errors represented by quantum jump happens independently in the level-1 gadgets. Consider a level-1 circuit composed of many $t$-FT level-1 gadgets that can be grouped into extented rectangles (see e.g., Fig.~\ref{fig:FTreduce}(a)). 
    If there are at most $t$ quantum errors in each extended rectangle, we can convert it to an ideal gadget. In that case, only when more than $t$ errors occur in the same extended rectangles at the same time can one logical error happen, which owns a probability of $O(p^{t+1})$.
\end{proof}

In the following, we focus on constructing FT level-$1$ bosonic gadgets that satisfy the above definitions by integrating bosonic quantum error correction and quantum controls. More specifically, given a bosonic code $\mc{C}$ that can correct loss and phase-rotation errors, e.g. the cat code, we try to design error-corrected $\mc{C}$-encoded gadgets by carefully engineering the Hamiltonian of their composing AAOs so that physical faults propagate controllably to data errors. An analogous example in the context of qubit fault-tolerance is the use of transversal gates~\cite{shor1996fault}, which guarantees that a single circuit fault only propagates to a single data error per code block. However, this quantum control task is more sophisticated when involving bosonic modes as we need to consider complicated continuous evolution in a larger Hilbert space. 

In order for a level-$1$ gadget to be FT, it has to tolerate both bosonic faults and ancilla faults. Tolerance against bosonic faults can be achieved relatively easily by using the error-transparency control~\cite{ma_error-transparent_2020}, or more generally, the error-closure control~\cite{tsunoda2022error}. Tolerance against ancilla errors is usually harder to achieve since some DV ancilla errors tend to propagate uncontrollably and a small ancilla fault could lead to a catastrophic error in the bosonic mode. Fortunately, path-independent control~\cite{ma2020path, ma_algebraic_2022, reinhold2020error} has been proposed for controlling the ancilla faults propagation to the bosonic mode. 
However, the previously defined PI condition~\cite{ma2020path} is more stringent than what is required by fault tolerance. Thus in the next section, we will generalize the PI control, relax its requirement, and rigorously connect its generalized version to fault-tolerance analyzed in this section. 

\section{Generalized path-independent operations} \label{Sec:GeneralizedPI}
We first review the PI control proposed in Ref.~\cite{ma2020path}. Again, we consider a Markovian interaction between a bosonic mode $C$ and a $d \geq 2$-level ancilla $A$ described by Eq.~\eqref{eq:Lindblad_equation}, where only the ancilla noises are considered, i.e.
\begin{equation}
    \frac{d \rho}{dt} = -i [H_{AC}(t), \rho] + \sum_j \mathcal{D}[\sqrt{\gamma_j} J_j] \rho
    \label{eq:master_equation_ancillla_errors_only}
\end{equation}
where $J_j$ are some jump operators acting only on the ancillary system. 
The ancilla is initialized in some initial state $\ket{i}_A$, and measured in a certain basis $\{\ket{r}_A\}$ after an interaction time $T$. Let $\mc{G}(T)$ denote the joint channel generated by the Lindblad master equation in Eq.~\eqref{eq:master_equation_ancillla_errors_only} for a time $T$. With a slight abuse of notation, we may neglect the subscripts $A$ or $C$ labeling the ancilla or the bosonic mode for states or operators without ambiguity. We denote $\bbra r \mc{G} \kket{i} := \bra{r} \mc{G}(\ket{i}\bra{i}\otimes \bullet)\ket{r}$ as the (unnormalized) channel on the bosonic mode conditioned on the initial and final ancilla state $\ket{i}$ and $\ket{r}$~\cite{ma_algebraic_2022}. 
A PI gate is defined as follows.
\begin{definition}[PI gate]
An ancilla-assisted gate $\mathcal{G}(T)$ is PI in an ancilla basis $\mc{B}_A$ with an ancilla initial state $\ket{i}$ if for any $\ket{r} \in \mc{B}_A$,
\begin{equation}
     \langle\langle r| \mathcal{G}(T)|i\rangle\rangle \propto U_{ri}\bullet U_{ri}^{\dagger},
    \label{eq:PI_condition}
\end{equation}
where $U_{ri}$ is some $r$-dependent unitary on the bosonic mode. 
\label{def:PI_gate}
\end{definition}
The PI condition in Eq.~\eqref{eq:PI_condition} implies that each conditional channel does not contain any errors (it is a unitary channel without information loss) propagated from the ancilla, although the unconditional channel might. In other words, no information is lost to the environment by performing such a noisy operation if the ancilla measurement outcome is accessible. See Fig.~\ref{fig:illustration_ancilla_assisted_operation} for an illustration of such a PI gate.

Note that the PI condition in Eq.~\eqref{eq:PI_condition} is for the joint channel, which could be hard to evaluate directly. As such, Ref.~\cite{ma_algebraic_2022} provided an easy-to-evaluate algebraic condition for the Hamiltonian $H_{AC}(t)$ and the jump operators $\{J_j\}$ in order for $\mc{G}$ to satisfy Eq.~\eqref{eq:PI_condition}, which we present in Appendix~\ref{appendix:algebric_conditions}.

The PI definition in Def.~\ref{def:PI_gate} considers an infinite number of ancilla-faults since it examines the full $\mathcal{G}(T)$. In practice, when the ancilla noises are small, by correcting only a finite number of ancillary faults, we can suppress the real logical error rates to a higher order. As such, we define the following finite-order PI gate that only concerns a finite truncation of $\mc{G}(T)$:

\begin{definition}[Finite-order PI gate] 
An ancilla-assisted gate is $n$-PI in an ancilla basis $\mc{B}_A$ with an ancilla initial state $\ket{i}$ if for any $\ket{r} \in \mc{B}_A$ and any $k\leq n$, 
\begin{equation}
    \langle\langle r| 
    \mathcal{G}^{[k]}(T) |i\rangle\rangle \propto U_{ri}\bullet U_{ri}^{\dagger},
    \label{eq:def_finite_PI}
\end{equation}
where $U_{ri}$ is some $r$-dependent unitary on the bosonic mode. 
\label{def:finite-order_PI_gate}
\end{definition}

\begin{figure}
    \centering
    \includegraphics[width=0.5\textwidth]{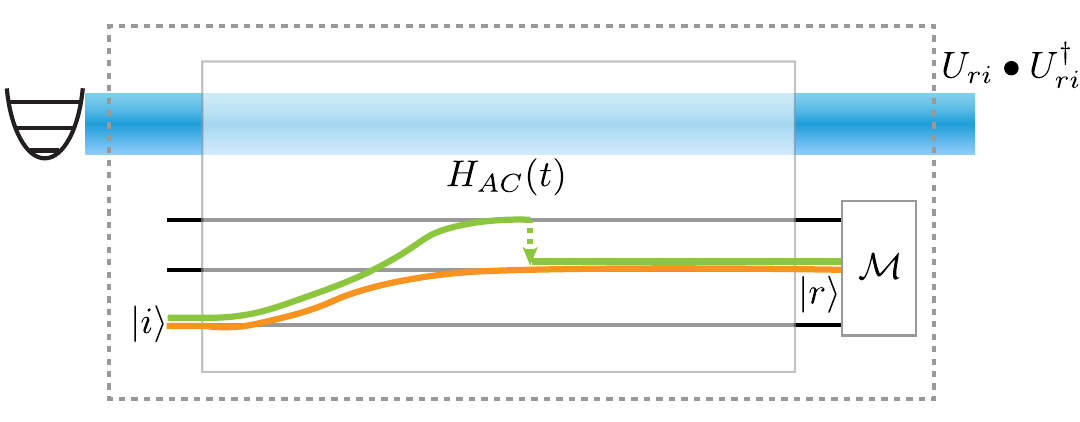}
    \caption{Illustration of a PI gate. Given an ancilla initial state $\ket{i}$ and a measurement basis $\mc{B}_A$, the bosonic mode undergoes a $r$-dependent unitary $U_{ri}$ for any ancilla measurement outcome $\ket{r}\in \mc{B}_A$, independent of the different system paths (see e.g. the green and orange curves, where an ancilla relaxation happens for the green curve).}
    \label{fig:illustration_ancilla_assisted_operation}
\end{figure}

In Appendix~\ref{appendix:algebric_conditions}, we present an algebraic condition for the Hamiltonian and jump operators in order for $\mc{G}$ to satisfy Eq.~\eqref{eq:def_finite_PI}.

The PI condition, even with a finite order, is demanding since it requires the conditional channels to be exactly unitary channels and thus allows no error propagation at all. However, observe that if the bosonic mode is protected by some bosonic codes, fault-tolerance can still be achieved even if we allow error propagations, as long as the propagated errors are small and correctable. Based on this observation, we generalize the PI control and present a less stringent condition that, nevertheless, still complies with the idea of fault-tolerance:
\begin{definition} [GPI operation]
Given a single-mode bosonic code with a codespace projection $P_c$, we say that an ancilla-assisted operation is $n$-th order generalized path-independent (GPI) in an ancilla basis $\mc{B}_A$ with an initial ancilla state $\ket{i}$ if for any $\ket{r}\in \mc{B}_A$ and $k \leq n$,
\begin{equation}
    \langle\langle r| \mathcal{G}^{[k]}(T) |i\rangle\rangle \propto (\sum_s K^s_{ri}\bullet K^{s \dagger}_{ri}),
    \label{eq:GPI_operation}
\end{equation}
where $\{K^s_{ri}\}_s$ satisfies the KL condition, i.e. $P_c K^{s \dagger}_{ri} K^{s^{\prime}}_{ri}P_c \propto P_c$ for any $K_{ri}^s, K_{ri}^{s^{\prime}} \in \{K^s_{ri}\}_s$.
\label{def:GPI_operation}
\end{definition}
Note that any conditional channel $\langle\langle r| \mathcal{G}^{[k]}(T) |i\rangle\rangle$ can be written in the form of Eq.~\eqref{eq:GPI_operation}, with a set of $(r,i)$-dependent Kraus operators $\{K^s_{ri}\}_s$. The condition that $\{K^s_{ri}\}_s$ satisfies the KL condition implies that the conditional channel $\langle\langle r| \mathcal{G}^{[k]}(T) |i\rangle\rangle$ contains only correctable errors. 

The GPI condition generalizes from the PI condition in Def.~\ref{def:finite-order_PI_gate} from the following two aspects. First, the GPI condition considers any operation (any CPTP map) to the bosonic mode as a target, while the PI condition only considers unitary gates. Second, the GPI condition allows finite propagation of ancilla faults to the bosonic mode for each conditional channel, as long as the propagated errors are correctable by the bosonic code. See Fig.~\ref{fig:gadget}(b) for an illustration of the relation between ancilla-assisted operations, GPI operations and PI operations.

In Appendix~\ref{appendix:algebric_conditions}, we present an algebraic condition for GPI operations again by only examining the Hamiltonian and jump operators. Note that we directly present the GPI condition in the finite-order form, which is of practical relevance. 




\subsection{GPI examples}
Here, we provide two examples of GPI operations for the four-legged cat code. 
\subsubsection{GPI SNAP gate with a three-level $\chi$-mismatched ancilla \label{sec:GPI_SNAP}}
As an example, we consider the photon-number selective phase (SNAP) gate~\cite{heeres2015cavity} in circuit-QED systems. In the rotating frame, a three-level ancilla with levels $\{\ket{g}, \ket{e}, \ket{f}\}$ is dispersively coupled to a bosonic mode via the Hamiltonian 
\begin{equation}
    H_0 = - (\chi_f \ket{f}\bra{f} + \chi_e \ket{e}\bra{e})\otimes a^{\dagger}a,
    \label{eq:SNAP_Hamiltonian1}
\end{equation}
and the ancilla is frequency-selectively driven between $\ket{g}$ and $\ket{f}$ states:
\begin{equation}
    H_c(t) = \Omega \sum_{n = 0}^N e^{- i (n\chi_f t - \phi_n)}\ket{f}\bra{g} + h.c.,
    \label{eq:SNAP_Hamiltonian2}
\end{equation}
where $\vec{\phi} := (\phi_0, \phi_1, \cdots, \phi_N)$ is some phase vector that we can choose. We consider ancilla jump operators $\{J_1 = \sqrt{\gamma} \ket{e}\bra{f}, J_2 = \sqrt{\gamma} \ket{g}\bra{e}, J_3 = \sqrt{\gamma} \sum_{s \in \{e, f\}} \Delta_s \ket{s}\bra{s}\}$, where $J_1$ describes the ancilla decays from $\ket{f}$ to $\ket{e}$, $J_2$ the ancilla decay from $\ket{e}$ to $\ket{g}$, and $J_3$ an ancilla dephasing with arbitrary phases $\Delta_e, \Delta_f \in \mathbb{C}$. We will use this error model throughout the paper whenever using a three-level ancilla.  

In the interaction picture associated with $H_0$, the system Hamiltonian reads
\begin{equation} \label{eq:SNAP_interacting}
    \tilde{H} = \Omega \left[|f\rangle \langle g|\otimes S(\vec{\phi}) + h.c.\right],
\end{equation}
where $S(\vec{\phi}) := \sum_{n = 1}^N e^{i\phi_n} \ket{n}\bra{n}$ applies a number dependent phase shift $\vec{\phi}$ to the bosonic mode. Note that we have performed the rotating wave approximation by assuming $\Omega \ll \chi_f$. Denote $\Delta\chi:= \chi_f - \chi_e$ as the $\chi$ mismatch. The ancilla jump operators transform to 
$\tilde{J}_1(t) = \sqrt{\gamma}|e\rangle \langle f| \otimes e^{-i\Delta \chi t a^{\dagger}a}, \tilde{J}_2(t) = \sqrt{\gamma}|g\rangle \langle e| \otimes e^{-i \chi_e t a^{\dagger}a}$, and $\tilde{J}_3 = J_3$. 

We initialize the ancilla in $\ket{g}$ and let the system evolve for a time $T = \pi/2\Omega$, and measure the ancilla in the $\{\ket{g}, \ket{e}, \ket{f}\}$ basis. In the absence of errors, the ancilla will end up in $|f\rangle$ while the central system undergoes the target number-dependent phase shifts $S(\vec{\phi})$, i.e. $\langle f| e^{-i\tilde{H}_c T}|g\rangle = S(\vec{\phi})$. 
With ancilla errors, we can explicitly write down the truncated conditional channels (in the interaction picture) up to the first order:
\begin{equation}
    \begin{aligned}
        \langle\langle g|\tilde{\mathcal{G}}^{[1]}(T)|g\rangle\rangle & \propto \mathcal{I}, \\
        \langle\langle f|\tilde{\mathcal{G}}^{[1]}(T)|g\rangle\rangle & \propto S(\vec{\phi})\bullet S^{\dagger}(\vec{\phi}), \\
        \langle\langle e|\tilde{\mathcal{G}}^{[1]}(T)|g\rangle\rangle & \propto \int_{t=0}^T dt\ e^{-i\Delta \chi t a^{\dagger}a} S(\vec{\phi})\bullet S^{\dagger}(\vec{\phi}) e^{i\Delta \chi t a^{\dagger}a},\\
    \end{aligned}
    \label{eq:channel_GPI_SNAP}
\end{equation}
If there is no $\chi$-mismatch, i.e. $\Delta \chi = 0$, this gate is strictly a $1$-PI gate (see Eq.~\eqref{eq:def_finite_PI}); If there is a finite $\chi$-mismatch, the gate is no longer PI. Nevertheless, for a bosonic code that can correct phase rotations in the range $[-\theta_{m}/2, \theta_{m}/2]$ (e.g. $\theta_{m} = \pi/2$ for the four-legged cat), the gate is still a $1$-GPI gate if $\Delta \chi T \leq \theta_{m}$ (see Eq.~\eqref{eq:GPI_operation}).

In Appendix~\ref{appendix:algebric_conditions}, we show that one can verify the GPI property of this SNAP gate more easily without calculating the conditional channels by checking a set of algebraic conditions for the Hamiltonian and jump operators. Also, in Appendix~\ref{appendix:GPI_parity_flag}, we present another GPI SNAP scheme using a qutrit and an extra flag qubit, which can tolerate even larger $\chi$-mismatch $\Delta \chi T$. 
\subsubsection{GPI parity measurement with a three-level $\chi$-mismatched ancilla
\label{sec:GPI_parity_measurement}}
As another example of GPI operations, we consider the parity measurement for correcting photon loss errors~\cite{rosenblum2018fault} using a three-level $\chi$-mismatched ancilla.

The system Hamiltonian (in the rotating frame) is $H_0 = - (\chi_e |e\rangle \langle e| + \chi_f |f\rangle \langle f|)\otimes a^{\dagger}a$ (without ancilla drives). We denote $\ket{\pm}$ as $(\ket{g} \pm \ket{f})/\sqrt{2}$. The ancilla is initialized in $\ket{+}$ and measured in the basis $\{\ket{+}, \ket{-}, \ket{e}\}$. 

In the absence of ancilla errors, 
the operation performs a projection onto the even (odd) subspace of the bosonic mode conditioned on the ancilla measured in $\ket{+}$ ($\ket{-}$):
\begin{equation}
\begin{aligned}
    \bbra{+}\mc{G}^{[0]}\kket{+} & = P_+ \bullet P_+, \\
    \bbra{-}\mc{G}^{[0]}\kket{+} & = P_- \bullet P_-, \\
\end{aligned}
\end{equation}
where $P_{\pm} := (I \pm e^{-i\pi a^{\dagger}a})/2$ is the projection on the even/odd parity subspace of the bosonic mode.

In the presence of ancilla errors $\{J_1 = \sqrt{\gamma}\ket{e}\bra{f}$, $J_2 = \sqrt{\gamma}\ket{g}\bra{e}$, $J_3 = \sqrt{\gamma}\sum_{s \in \{e, f\}} \Delta_s \ket{s}\bra{s}\}$, we move to the interaction picture associated with $H_0$.
Now, the system Hamiltonian is $0$ and the ancilla jump operators read $\tilde{J}_1(t) = \sqrt{\gamma}|e\rangle \langle f|\otimes e^{-i\Delta \chi t a^{\dagger}a}, \tilde{J}_2(t) = \sqrt{\gamma} |g\rangle \langle e|\otimes e^{-i\chi_e t a^{\dagger}a}$, and $\tilde{J}_3 = J_3 = \sqrt{\gamma}\sum_{s \in \{e, f\}} \Delta_s \ket{s}\bra{s}\}$, same to those in the previous SNAP example. Here, without loss of generality, we set $\Delta_f = -1$. 

We can calculate the the noise expansion of the joint channel up to the first order (see Eq.~\eqref{eq:dyson_expansion_Kraus_form}):
\begin{equation}
\begin{aligned}
    \tilde{\mathcal{G}}^{[1]}(T) & = W(T,0)\bullet W^{\dagger}(T,0) + \gamma \int_{t = 0}^T G_1(t, 1)\bullet G_1^{\dagger}(t, 1)\\
    & + \gamma \int_{t = 0}^T G_1(t, 3)\bullet G_1^{\dagger}(t, 3),
\end{aligned}
\end{equation}
where $W(t_2, t_1) := \exp\left[-i H_{\mr{eff}}(t_2 - t_1)\right]$ with $H_{\mr{eff}} := -\frac{i}{2}\sum_{j=1}^3 \tilde{J}_j^{\dagger}\tilde{J}_j = - \frac{i}{2}\gamma [(1 + |\Delta_e|^2)\ket{e}\bra{e} + 2 \ket{f}\bra{f}]$, and
$G_1(t, j):= W(T, t)\tilde{J}_j(t)W(t,0)$. 
Note that we have dropped the term associated with the first-order quantum jump with $\tilde{J}_2$, which is zero when the ancilla starts from $\ket{+}$. 
Going back to the lab frame, the truncated channel is $\mathcal{G}^{[1]}(T) = \tilde{\mathcal{G}}^{[1]}(T)\circ [U_0(T)\bullet U_0^{\dagger}(T)]$, where $U_0(T) := e^{-i H_0 T}$. Then we can calculate the truncated conditional channels:
\begin{equation}
\begin{aligned}
    \langle\langle +|\mathcal{G}^{[1]}|+\rangle\rangle & = [(1-\frac{p}{2})P_+ + \frac{p}{2}P_-]\bullet [(1-\frac{p}{2})P_+ + \frac{p}{2}P_-], \\
    & + p P_- \bullet P_- + O(p^2)\\
    \langle\langle -|\mathcal{G}^{[1]}|+\rangle\rangle & = [(1-\frac{p}{2})P_- + \frac{p}{2}P_+]\bullet [(1-\frac{p}{2})P_- + \frac{p}{2}P_+], \\ 
    & + p P_+ \bullet P_+  + O(p^2)\\
    \langle\langle e|\mathcal{G}^{[1]}|+\rangle\rangle & = \frac{p}{2T} \int_{t=0}^T dt e^{-i(\Delta\chi t + \pi) a^{\dagger}a}\bullet e^{i(\Delta\chi t + \pi) a^{\dagger}a}\\
    & + O(p^2),\\
    \label{eq:conditional_channels_parity}
\end{aligned}
\end{equation}
where $p := \gamma T$. For a four-legged cat with $\alpha \gg 1$, Eq.~\eqref{eq:conditional_channels_parity} satisfies the GPI condition as long as $\Delta \chi T < \pi/2$. Note that the first two terms in Eq.~\eqref{eq:conditional_channels_parity} imply that one might obtain wrong parity measurement outcomes with a probability $O(p)$ if the ancilla is measured in $\ket{\pm}$. Such effective measurement errors can be suppressed to the second order by repeating the above parity measurement three times and performing majority voting, which will be discussed in the next section when we rigorously consider fault-tolerance. 


\section{Connection between GPI and fault-tolerance} \label{Sec:GPItoFT}
In this section, we establish the connection between GPI quantum control and fault-tolerance defined in Sec.~\ref{sec:fault_tolerance}. Let the bosonic mode be encoded in some bosonic code with a code projection $P_c$. 
\begin{proposition}
    Each AAO contained in a $t$-FT level-$1$ gadget with an ancilla initial state $|i\rangle$ and an ancilla measurement basis $\mc{B}_A$ has to be $t$-GPI with respect to $|i\rangle$, $\mc{B}_A$, and the code projection $P_c$. 
\end{proposition}
\begin{proof}
Any $t$-FT gadget requires that if any $m \leq t$ faults occur during the protocol, the output is guaranteed to be correct. However, if one AAO is not $t$-GPI, there exists an ancilla measurement outcome $r$, conditioned on which the bosonic channel (see Eq.~\eqref{eq:GPI_operation}) contains non-correctable errors. As a result, the final output can no longer be guaranteed to be correct.
\end{proof}


Conversely, we can combine pieces of $t$-GPI operations to get a $t^{\prime} \leq t$-FT gadgets, as shown in Fig.~\ref{fig:gadget}. In order to make $t^{\prime} = t$, there are extra requirements for the AAOs, which are typically easier to satisfy than GPI. Instead of making rigorous statements about these requirements for generic gadgets, we will make case studies when constructing concrete FT gadgets. 
Nevertheless, we comment on some high-level ideas used to design the extra ingredients that can be combined with GPI to achieve fault tolerance here: (i) Operations are error transparent/closure against bosonic errors (see Sec.~\ref{appendix:error_transparency}); (ii) The propagation from ancilla faults to bosonic errors is linear; (iii) There exists at least one ancilla state $\ket{r}$ such that the ideal conditional bosonic channel $\bbra{r}\mc{G}^{[0]}\kket{i}$ gives the target operation. 

As the first example, we construct $1$-FT Z-axis rotation $Z(\theta)$ for the four-legged cat using the GPI SNAP gate presented in Sec.~\ref{sec:GPI_SNAP}. To implement a $Z(\theta)$ gate, we choose a GPI SNAP gate with $\Delta \chi T < \pi/2$ and 
\begin{equation}
    S(\vec{\theta}) = P_0 + P_3 + e^{i\theta}(P_2 + P_1),
    \label{eq:ET_SNAP}
\end{equation}
where $P_j := \sum_{i=0}\ket{4i + j}\bra{4i + j}$. We consider the same ancilla jump errors as those presented in Sec.~\ref{sec:GPI_SNAP}. In addition, we consider a single jump operator $a$ representing a single photon loss for the bosonic mode.
Then we implement the $1$-FT $Z(\theta)$ gate based on Algorithm~\ref{alg:Z-rotation} below. The $3$-level ancilla basis is denoted by $\ket{g}$,$\ket{e}$ and $\ket{f}$ according to Eq.~\eqref{eq:SNAP_Hamiltonian1}.
\begin{algorithm}
$o \gets e$. 
\tcp{$o$ records the ancilla measurement outcome} 
\While{$o \neq f$}{
Prepare the anilla in the $|g\rangle$ state, apply the GPI SNAP gate with $S(\vec{\theta})$ in Eq.~\eqref{eq:ET_SNAP} for a time $T = \pi/2\Omega$, and measure the ancilla in the $|g\rangle, |e\rangle, |f\rangle$ basis with an outcome $o \in \{g, e, f\}$. \\
\If{$o = e$}{Apply a phase rotation $e^{i\Delta \chi T a^{\dagger}a/2}$ to the bosonic mode.}
}
\caption{$1$-FT $Z(\theta)$ gate}
\label{alg:Z-rotation}
\end{algorithm}

Now, we verify that the above protocol satisfies the definition of a $1$-FT gate in Def.~\ref{def:FT_gates}. Here, we choose $f(m) = (m, m \Delta \chi T/2)$ with $\Delta\chi T/2 < \pi/4$.  Suppose the input error is of size $(k, \theta_0)$ and there are $m$ faults during the protocol. There are two cases where $(k, \theta_0) + f(m) \leq f(1)$. First, $m = 0$ and $(k, \theta_0) \leq (1, \Delta\chi T/2)$. Obviously, the gate is error-transparent to the phase rotation $e^{-i\theta a^{\dagger}a}$, i.e. it simply commutes through the gate and remains at the output, since it commutes with the system Hamiltonian (see Eq.~\eqref{eq:SNAP_Hamiltonian1} and \eqref{eq:SNAP_Hamiltonian2}). Moreover, as we shown in Appendix~\ref{sec:ET_SNAP}, the gate is also error-transparent to a single photon loss $a$ when using the form of $S(\vec{\phi})$ in Eq.~\eqref{eq:ET_SNAP}. Therefore, the input $(k \leq 1, \theta \leq \Delta \chi T/2)$ error simply remains at the output and stays correctable.
Second, $m = 1$ and $(k, \theta) = (0,0)$. In this case, either an ancilla dephasing, or an ancilla decay, or a single photon loss occurs during the protocol.
A single ancilla dephasing might cause the ancilla ending in $\ket{g}$ instead of $\ket{f}$ but does not add any error to the bosonic mode; A single ancilla decay from $\ket{f}$ to $\ket{e}$ only causes a correctable phase rotation with an angle $|\delta \theta| \leq \Delta \chi T/2 < \pi/4$~\footnote{An ancilla decay from $\ket{f}$ to $\ket{e}$ induces phase rotations in the range of $[0, \Delta\chi T]$ (see Eq.~\eqref{eq:channel_GPI_SNAP}). The feedback rotation by $-\Delta \chi T/2$ then shifts the phase rotation window to $[-\Delta\chi T/2, \Delta \chi T/2]$}; A single-photon loss simply remains at the output and stays correctable.  

As the second example, we construct a $1$-FT QEC protocol for correcting a single-photon loss. Note that we will present a full EC gadget correcting both photon loss and dephasing errors in the next section. The protocol utilizes the $1$-GPI parity measurement presented in Sec.~\ref{sec:GPI_parity_measurement}, with a $\chi$ mismatch $\Delta \chi T < \pi/2$. 
\begin{algorithm}[ht!]
$o_{i} \gets e$ for $i \in \{1,2,3\}$.
\tcp{$\{o_i\}_{i \in [3]}$ record three consecutive parity measurement outcomes}
\For{$i \gets 1\; \KwTo\; 3$}{
\While{$o_i = e$}{
Prepare the anilla in the $\ket{+}$ state, apply the dispersive coupling for a time $T = \pi/\chi_f$, and measure the ancilla in the $=\{\ket{+},\ket{-}, \ket{e}\}$ basis with an measurement outcome $o_i$.\\
\If{$o_i = e$}{Apply a phase rotation $e^{i\Delta \chi T a^{\dagger}a/2}$ to the bosonic mode.}
}
}
Apply a parity correction based on the majority voting over $\{o_i\}_{i \in [3]}$.
\caption{$1$-FT photon-loss correction}
\label{alg:FT_QEC}
\end{algorithm}

Now, we verify that the protocol in Alg.~\ref{alg:FT_QEC} satisfies the definition of a $1$-FT QEC in Def.~\ref{def:FT_QEC}. Similar to the previous $Z(\theta)$ gate example, we choose $f(m) = (m, m\Delta\chi T/2)$. Assume there is an input error of size $(k, 0)$ and $m$ faults occur during the protocol. Note that since we are only correcting single photon losses for now, we assume the input has no dephasing errors. 
To verify condition (i) in Def.~\ref{def:FT_QEC}, we consider either $k = 1$, $m = 0$ or $k = 0$, $m = 1$. In the earlier case, the single photon loss can be perfectly corrected and the output has no error; in the latter case, we consider either an ancilla dephasing, an ancilla decay, or a single photon loss. An ancilla dephasing may flip a single parity measurement outcome but does not affect the final majority voting; A single ancilla decay only causes a correctable phase rotation with amplitude $\leq \Delta \chi T/2 < \pi/4$, which is a correctable error; A single photon loss during the protocol either gets corrected or goes undetected but remains as a correctable single photon loss at the output. 

For condition (ii) in Definition~\ref{def:FT_QEC}, one simply observes that a single photon loss error at the input can be detected and then corrected (although a logical error may happen when combined with another photon loss during the protocol), while a single photon loss or an ancilla decay can cause at most a $f(1) = (1, \Delta \chi T/2)$ error that can go undetected and remains at the output. 

\section{Fault-tolerant operations of four-legged cat code} \label{sec:FT4legCat}
In this section, we focus on the four-legged cat and construct universal $1$-FT level-$1$ gadgets that can correct a single-photon loss and any single ancilla fault, using GPI operations.

The universal operation set we consider is 
\begin{equation}
    \mathcal{S} = \{\mr{EC}, Z(\theta), X(\phi), XX(\delta), \mathcal{P}_{\ket{+_L}}, \mathcal{M}_Z, \mathcal{M}_X\},
\end{equation}
including error correction, $Z$-axis rotation, $X$-axis rotation, $XX$ rotation ($\exp(-i\delta XX /2)$), state preparation in the $X$ basis, measurement in the $Z$ basis, and measurement in the $X$ basis.

One essential element for our construction is the GPI SNAP gate and GPI parity measurement described in Sec.~\ref{sec:GPI_SNAP} and Sec.~\ref{sec:GPI_parity_measurement}, respectively. Recall that both of these two operations use a three-level ancilla, which is dispersive coupled to the bosonic mode via $-(\chi_e \ket{e}\bra{e} + \chi_f \ket{f}\bra{f})\otimes a^{\dagger}a$, potentially with a $\chi$ mismatch $\Delta \chi := \chi_f - \chi_e$.
Denote the gate time for the SNAP gates as $T$ and that for a parity measurement as $T_P$. Typically $T \gg T_P$ since the driving strength $\Omega$ for the SNAP gate (see Eq.~\eqref{eq:SNAP_Hamiltonian2}) is much smaller than $\chi_f$ in order for the rotating-wave-approximation to hold~\cite{reinhold2020error}. We choose $f(m) = (m, m\Delta \chi T/2)$ with $\Delta\chi T/2 < \pi/8$~\footnote{Here, we choose $\Delta \chi T/2 < \pi/8$ instead of $\pi/4$ so that the relative phase shift between two data bosonic modes is smaller than $\pi/4$, which will be used to prove the fault tolerance of the two-qubit logical gate ($XX$ rotation)} for proving the fault-tolerance of the gadgets. Unless specially noted, all the SNAP gates and parity measurement gadgets we use have a $\chi$ mismatch $\Delta \chi$. 

Similar to the previous sections, we consider $\{a, \ket{e}\bra{f}, \ket{g}\bra{e}, \sum_{s \in \{e,f\}}\Delta_s \ket{s}\bra{s}\}$ as the errors, representing a single photon loss, an ancilla decay from $\ket{f}$ to $\ket{e}$, an ancilla decay from $\ket{e}$ to $\ket{g}$, and an ancilla  dephasing, respectively.

\subsection{Z-axis rotation}
A $1$-FT $Z$-axis rotation with an arbitrary angle $(\theta)$ using GPI SNAP gate is presented in Alg.~\ref{alg:Z-rotation} in the previous section. Note that a $1$-FT logical gate using strictly PI SNAP gate (with no $\chi$ mismatch) has been experimentally demonstrated for a small binomial bosonic code~\cite{reinhold2020error}. Here, the main difference is that our protocol allows a finite $\chi$ mismatch.

\subsection{X-axis rotation}
In the large $\alpha$ limit, a $X$-axis rotation is given by 
\begin{equation}
    X(\phi) \approx e^{i\phi}\ket{C^+_{\alpha}}\bra{C^+_{\alpha}} + \ket{C^+_{i\alpha}}\bra{C^+_{i\alpha}},
\end{equation}
where $\ket{C^{\pm}_{\beta}}:= c^{\pm}_{\beta}(\ket{\beta} \pm \ket{-\beta})$, with $c^{\pm}_{\beta}$ being normalization constants. We implement $X(\phi)$ by adding a phase $\phi$ to the subspace spanned by the two coherent states $\ket{\alpha}$ and $\ket{-\alpha}$. As illustrated in Fig.~\ref{fig:level_1_gadgets}(a), we first displace the cavity by $\alpha$ and apply a phase shift to the vacuum $S(\vec{\phi}) = e^{i\phi}\ket{0}\bra{0} + I - \ket{0}\bra{0}$ using the SNAP gate (see Sec.~\ref{sec:GPI_SNAP}). Then we displace the cavity by $-2\alpha$ and apply another $S$. Finally, we displace the cavity by $\alpha$ back to the codespace. The joint operation is:
\begin{equation}
\begin{aligned}
    U_X & = D(\alpha) S(\vec{\phi}) D(-2\alpha) S(\vec{\phi}) D(\alpha) \\
    &= [D(\alpha)S(\vec{\phi}) D(\alpha)^{\dagger}][D(-\alpha)S(\vec{\phi}) D(\alpha)^{\dagger}] \\
    & \approx e^{i\theta}P_{\pm \alpha} + I - P_{\pm \alpha},
    \label{eq:X-rotation}
\end{aligned}
\end{equation}
where $P_{\pm \alpha} := \ket{\alpha}\bra{\alpha} + \ket{-\alpha}\bra{-\alpha} = \ket{C^+_{\alpha}}\bra{C^+_{\alpha}} + \ket{C^-_{\alpha}}\bra{C^-_{\alpha}}$.

We now show that this gate is $1$-FT if the $\chi$-mismatch during the SNAP gates is zero. Assume there is a $(k,\delta\theta)$ input error and $m$ faults occur during the gate.
Again, for $1$-FT gate (see Def.~\ref{def:FT_gates}), we only need to consider either $(k = 0,\delta \theta = 0)$, $m = 1$, or  $(k \leq 1,\delta\theta \leq \Delta\chi T/2)$, $m = 0$. 

First, we consider a single fault occurring during $U_X$. A single-photon loss simply commutes through the entire gate since the two SNAP gates $S$ are error-transparent (see Appendix~\ref{appendix:error_transparency}) and $D(\alpha)$ commutes with $a$ up to a constant. A single-ancilla decay or dephasing during the $S$ gate does not cause any error to the bosonic mode when assuming perfect $\chi$ matching. Therefore, a single fault during the gate causes at most a $(1,0) < f(1)$-error at the output, which is correctable. 

Second, we consider a $(k \leq 1,\delta \theta \leq \Delta\chi T/2)$ input error $e^{i\delta\theta a^{\dagger}a} a^k$. We first notice that $U_X e^{i\delta\theta a^{\dagger}a}a^k P_c \propto a^k U_X e^{i\delta\theta a^{\dagger}a} P_c$ since $U_X$ is error-transparent for $a^k$ (see Eq.~\eqref{eq:X-rotation}). Here, $P_c:= \ket{+_L}\bra{+_L} + \ket{-_L}\bra{-_L} \approx \ket{C_\alpha^+}\bra{C_\alpha^+} + \ket{C_{i\alpha}^+}\bra{C_{i\alpha}^+}$ is the projector onto the code space of the four-legged cat. Then we only need to make sure that $U_X$ is also error-transparent to dephasing $e^{i\delta\theta a^{\dagger}a}$.  Let $E := U_X e^{i\delta\theta a^{\dagger}a} U_X^{\dagger}$ be the effective error that $e^{i\delta\theta a^{\dagger}a}$ propagates to after the gate. $E$ satisfies
\begin{equation}
\begin{aligned}
    E P_c & = e^{i\delta\theta a^{\dagger}a} P_c + (1 - e^{-i\phi})(P_{\pm \alpha} - I)e^{i \delta\theta a^{\dagger}a}\ket{C^+_{\alpha}}\bra{C^+_{\alpha}} \\
    & + (e^{i\phi} - 1)P_{\pm \alpha} e^{i \delta\theta a^{\dagger}a} \ket{C^+_{i\alpha}}\bra{C^+_{i\alpha}},
    \label{eq:propagated_dephasing}
\end{aligned}
\end{equation}
where we can see that $U_X$ is not error-transparent against the dephasing due to the last two terms of Eq.~\eqref{eq:propagated_dephasing}. Fortunately, we can make it approximately error-transparent by modifying the SNAP gate:
\begin{equation}
    S(\vec{\phi}) \rightarrow e^{i\phi}(P_{[s]}) + I - P_{[s]},
\end{equation}
where $P_{[s]}:= \sum_{i=0}^s \ket{i}\bra{i}$ is the projection onto the $s$-neighborhood of vacuum.
Then the gate unitary becomes $U_X \rightarrow e^{i\phi}P_{\pm \alpha, s} + I - P_{\pm \alpha, s}$, where $P_{\pm\alpha, s} := D(\alpha)P_{[s]}D^{\dagger}(\alpha) + D(-\alpha)P_{[s]}D^{\dagger}(-\alpha)$ is the projection onto a neighborhood of $\ket{\alpha}$ and $\ket{-\alpha}$. Now, the effective error for the dephasing error becomes
\begin{equation} \label{eq:ErrorTransXrot}
\begin{aligned}
    E P_c & = e^{i\delta\theta a^{\dagger}a} P_c + (1 - e^{-i\phi})(P_{\pm \alpha, s} - I)e^{i \delta\theta a^{\dagger}a}\ket{C^+_{\alpha}}\bra{C^+_{\alpha}} \\
    & + (e^{i\phi} - 1)P_{\pm \alpha, s} e^{i \delta\theta a^{\dagger}a} \ket{C^+_{i\alpha}}\bra{C^+_{i\alpha}}.
\end{aligned}
\end{equation}
For $|\delta\theta| \leq \Delta\chi T/2 < \pi/8$, we can choose $s=O(|\alpha|^2)$ such that the last two terms vanish in the $\alpha \gg 1$ limit, i.e., 
\begin{equation}
    \begin{aligned}
        \langle C^+_{\alpha e^{i\delta \theta}}| P_{\pm\alpha, s} |C^+_{\alpha e^{i\delta \theta}}\rangle & \rightarrow 1, \\
        \langle C^+_{i\alpha e^{i\delta \theta}}| P_{\pm\alpha, s} |C^+_{i\alpha e^{i\delta \theta}}\rangle & \rightarrow 0. \\
    \end{aligned}
\end{equation}
Then we have $E P_c \approx e^{i\delta\theta a^{\dagger}a} P_c$ and the gate is error-transparent to dephasing as well.

Note that $1$-fault-tolerance can no longer be rigorously attained (even in the larger-$\alpha$ limit) if using SNAP gates $S$ with a finite $\chi$-mismatch. Taking the second $S$ gate as an example, and suppose it has a $\chi$-mismatch $\Delta \chi^{\prime}$, a single ancilla decay could cause a $e^{i\delta\theta^{\prime} a^{\dagger}a}$ phase rotation with $|\delta\theta^{\prime}| \leq \Delta\chi^{\prime}T/2$ after $S$, which propagates to $e^{-i\delta \theta^{\prime} [a^{\dagger}a + \alpha(a + a^{\dagger}) + \alpha^2]}$ after the later displacement. The extra displacement error after the gate is uncorrectable. Thus a single ancilla fault during the $X$-rotation can cause a first-order logical error with a probability $c p$, where $c$ is a constant depending on $\Delta \chi^{\prime} T$. Nevertheless, if $\Delta \chi^{\prime} T$ is small enough, the coefficient $c$ can be made comparable or even smaller than $p$, and we can still achieve good error suppression in practical regimes, as is shown in later numerical results in Fig.~\ref{fig:QEC_Infidelities}(a). 



\subsection{XX rotation \label{sec:XX_rotation}}
For large $\alpha$, the $XX$ rotation reads
\begin{equation}
\begin{aligned}
    XX(\delta) & \approx e^{i\delta} (\ket{C^+_{\alpha}, C^+_{\alpha}} \bra{C^+_{\alpha}, C^+_{\alpha}} + \ket{C^+_{i\alpha}, C^+_{i\alpha}} \bra{C^+_{i\alpha}, C^+_{i\alpha}}) \\
    & + (\ket{C^+_{\alpha}, C^+_{i\alpha}} \bra{C^+_{\alpha}, C^+_{i\alpha}} + \ket{C^+_{i\alpha}, C^+_{\alpha}} \bra{C^+_{i\alpha}, C^+_{\alpha}}).
\end{aligned}
\end{equation}
We implement $XX(\delta)$ by adding a phase $\delta$ to the subspace spanned by $\ket{\pm \alpha, \pm \alpha}$ and $\ket{\pm i\alpha, \pm i\alpha}$. As illustrated in Fig.~\ref{fig:level_1_gadgets}(b), we interfere the two modes through a $50:50$ beamsplitter, apply the number dependent phase shift $S(\vec{\delta}) = e^{-i\delta}\ket{0}\bra{0} + I - \ket{0}\bra{0}$ to both ports, and then interfere through another $50:50$ beamsplitter:
\begin{equation}
U_{XX} = \mathrm{BS}(\frac{\pi}{2})^{\dagger} (S\otimes S) \mathrm{BS}(\frac{\pi}{2}),
\label{eq:U_XX}
\end{equation}
where $\mathrm{BS}(\theta):= \exp [\frac{\theta}{2}(a b^{\dagger} - a^{\dagger}b)]$ with $a$ and $b$ denoting the annihilation operator of the two involved modes, respectively.  

To understand the effect of $U_{XX}$, we consider a coherent-state input $\ket{\alpha, \beta}$. The first BS interferes the two coherent states: 
\begin{equation}
    \mr{BS}\ket{\alpha, \beta}  = \ket{(\alpha + \beta)/\sqrt{2}, (\alpha - \beta)/\sqrt{2}}, 
    \label{eq:BS_effect}
\end{equation}
We take the approximation $S\ket{\gamma} \approx e^{-i\delta \mathbbm{1}[\gamma=0]}\ket{\gamma}$, where $\mathbbm{1}[x]$ is the indicator function. Then the two SNAP gates in Eq.~\eqref{eq:U_XX} add a nontrivial phase to the R.H.S. of Eq.~\eqref{eq:BS_effect} if $\alpha = \beta$ or $\alpha = - \beta$:
\begin{equation}
    (S\otimes S)\mr{BS}\ket{\alpha, \beta} = e^{-i\delta (\mathbbm{1}[\alpha=\beta] + \mathbbm{1}[\alpha=-\beta])} \ket{\frac{\alpha + \beta}{\sqrt{2}}, \frac{\alpha - \beta}{\sqrt{2}}}.
\end{equation}
Finally, the last BS restores the input coherent state potentially with an extra phase:
\begin{equation} \label{eq:UXXform}
   U_{XX} \ket{\alpha,\beta}= e^{-i\delta (\mathbbm{1}[\alpha=\beta] + \mathbbm{1}[\alpha=-\beta])} \ket{\alpha, \beta}.
\end{equation}
We remark that, when $\alpha$ and $\beta$ are only chosen from a set of discrete values $\{\alpha_i\}_i$ which are well-separated in the phase space, Eq.~\eqref{eq:UXXform} provides an exact expression of the action of $U_{XX}$. The rigorous form of $U_{XX}$ is given in Eq.~\eqref{eq:UXXform2} in Appendix~\ref{appendix:ET_XX_rotation}.
To conclude, a two-mode coherent state accumulates a nontrivial phase if and only if the two coherent states have matched amplitudes and aligned/anti-aligned phases. Let $P_{\pm(i)\alpha}$ be the projection onto a four-dimensional subspace spanned by $\ket{\alpha}, \ket{-\alpha}, \ket{i\alpha}, \ket{-i\alpha}$, we then have 
\begin{equation}
\begin{aligned}
& P_{\pm(i)\alpha}\otimes P_{\pm(i)\alpha} U_{XX} P_{\pm(i)\alpha}\otimes P_{\pm(i)\alpha} \\
& = e^{i\delta}(P_{\pm \alpha}\otimes P_{\pm \alpha} + P_{\pm i\alpha}\otimes P_{\pm i\alpha})\\
&\quad \quad  + ( P_{\pm\alpha}\otimes P_{\pm i\alpha} + P_{\pm i\alpha}\otimes P_{\pm\alpha}).
\end{aligned}
\label{eq:U_XX_alpha_projections}
\end{equation}
Note that Eq.~\eqref{eq:U_XX_alpha_projections} implies $P_c^{(AB)}  U_{XX} P_c^{(AB)} = XX(\theta)$ where $P_c^{AB} = P_c^{(A)}\otimes P_c^{(B)}$ is the projector onto the collective code space of 4-legged cat on bosonic modes $A$ and $B$. 

Now, we prove this $XX(\theta)$ gate is $1$-FT according to Def.~\ref{def:FT_gates}. We first consider the case where there is an input error $e^{i(\delta \theta_a a^{\dagger} a + \delta \theta_b b^{\dagger}b)}b^{k_b}a^{k_a}$ with $k_a, k_b \leq 1$ and $|\delta \theta_a|, |\delta \theta_b| \leq \Delta\chi T/2 < \pi/8$, but no fault during the gate. $b^{k_b}a^{k_a}$ simply commutes through the gate when acting on the code space due to the error-transparency form of $U_{XX}$ in Eq.~\eqref{eq:U_XX_alpha_projections}. Similar to proof for the $X$-axis rotation in Eq.~\eqref{eq:ErrorTransXrot}, one can show that $U_{XX}$ is also approximately error-transparent to the phase rotation $e^{i(\delta \theta_a a^{\dagger} a + \delta \theta_b b^{\dagger}b)}$ by replacing $S \rightarrow e^{-i\delta}(\sum_{i=0}^s \ket{i}\bra{i}) + I - (\sum_{i=0}^s \ket{i}\bra{i})$. We put the proof in Appendix~\ref{appendix:ET_XX_rotation}. As a result, the input error commutes through the gate and remains correctable.

To complete the proof that the $U_{XX}$ is $1$-FT, we also need to show that for a perfect input state and any single fault during $U_{XX}$, each of the two output modes has an error of size at most $f(1) = (1, \Delta\chi T/2)$. As shown in Appendix~\ref{appendix:ET_XX_rotation}, a single-photon loss during the gate propagates to an error of the form $c_1 a + c_2 b$, where $c_1, c_2 \in \mathbb{C}$, due to the error transparency of the SNAP gates and the error closure of the BSs. By using a $\chi$-matched ancilla for each SNAP gate, any single ancilla fault does not propagate to any bosonic data errors. 

We note that similar to the $X$-axis rotation, the $XX$ rotation is not strictly $1$-FT if there is a finite $\chi$-mismatch when executing the SNAP gates, as an induced phase rotation would propagate to uncorrectable errors after the last BS. Nevertheless, as we show numerically in Fig.~\ref{fig:QEC_Infidelities}, high-fidelity $XX$ rotation can still be realized in practical regimes even with a finite but small $\chi$-mismatch.

\subsection{X-basis state preparation}
The $+1$ $X$ basis eigenstate is a two-legged cat state with an even photon parity $\ket{+_L} = \ket{C^+_{\alpha}} = c^+_{\alpha} (\ket{\alpha} + \ket{-\alpha})$. Observe that $\ket{+_L} \propto P_+ \ket{\alpha}$, i.e. the even-parity projection of a coherent state $\ket{\alpha}$. Thus, we can prepare the even cat state by first preparing a coherent state $\ket{\alpha}$, and then performing a non-destructive parity measurement to project it to an even cat state (up to a parity frame update).
For $1$-FT state preparation, unlike the $1$-FT photon-loss correction protocol in Alg.~\ref{alg:FT_QEC}, we do not need to repeat the parity measurement three times as it allows a noisy output with up to $f(1) = (1, \Delta \chi T/2)$ error for up to a single fault during the parity measurement (see Def.~\ref{def:FT_state_preparation}). Concretely, we implement the following protocol:
\begin{algorithm}[ht!]
Prepare the bosonic mode in the coherent state $\ket{\alpha}$.\\
$o \gets e$ 
\tcp{records the parity measurement outcome}
\While{$o = e$}{
Prepare the ancilla in the $\ket{+}$ state, apply the dispersive coupling for a time $T = \pi/\chi_f$, and measure the ancilla in the $=\{\ket{+},\ket{-}, \ket{e}\}$ basis with an measurement outcome $o$.\\
\If{$o = e$}{Apply a phase rotation $e^{i\Delta \chi T a^{\dagger}a/2}$ to the bosonic mode.}
}
Apply a parity correction if $o = -$.
\caption{$1$-FT $X$-basis state preparation}
\label{alg:FT_X_preparation}
\end{algorithm}

Note that the $X$-basis state preparation here allows a finite $\chi$-mismatch.

\subsection{Z-basis measurement \label{sec:Z-basis_measurement}}
The $Z$-basis measurement admits the form of measuring photon number modulo 4. In order to obtain the correct logical measurement outcome in the presence of a single-photon loss, as required by Def.~\ref{def:FT_measurement}, we insert a non-destructive parity measurement before each logical $Z$ measurement. The full FT protocol is presented in Alg.~\ref{alg:FT_Z_measurement}.
\begin{algorithm}[ht!]
\For{$i \gets 1$ \KwTo $3$}{
$o_{i,a} \gets e$;\\
\While{$o_{i,a} = e$}{
Prepare the anilla in the $\ket{+}$ sate, apply the dispersive coupling for a time $T = \pi/\chi_f$, and measure the ancilla in the $\{\ket{+}, \ket{-}, \ket{e}\}$ basis with an measurement outcome $o_{i,a}$.}
$o_{i,b} \gets e$;\\
\While{$o_{i,b} = e$}{
\If{$o_{i,a} = +$}{
Prepare the anilla in the $\ket{+}$ sate, apply the dispersive coupling for a time $T =  \pi/2\chi_f$, and measure the ancilla in the $\{\ket{+}, \ket{-}, \ket{e}\}$ basis with an measurement outcome $o_{i,b}$.
}
\Else{Prepare the anilla in the $\ket{+}$ sate, apply the dispersive coupling for a time $T = \pi/2\chi_f$, apply an ancilla phase rotation $e^{-i\frac{\pi}{2}\ket{f}\bra{f}}$, and measure the ancilla in the $\{\ket{+}, \ket{-}, \ket{e}\}$ basis with an measurement outcome $o_{i,b}$.}
}
}
Obtain the logical measurement outcome as the majority voting of $\{o_{i,b}\}_{i=1,2,3}$. 
\caption{$1$-FT $Z$-basis measurement}
\label{alg:FT_Z_measurement}
\end{algorithm}

Note that each modulo-$4$ photon number measurement $o_{i,b}$ is conditioned on the parity measurement outcome $o_{i,a}$, i.e. we distinguish the photon number between $0\mod4$ and $2\mod4$ for even parity ($o_{i, b} = +$) and between $3\mod4$ and $2\mod4$ for odd parity ($o_{i,a} = -$). 

To verify that the $1$-FT measurement condition in Def.~\ref{def:FT_measurement} holds, one simply observe that a single photon loss before the measurement can be captured by the parity measurements, and any single fault during the measurement protocol can only cause at most one measurement error on one of $\{o_{i,b}\}_{i=1,2,3}$, and thus does not affect the majority voting. Note that the $Z$-basis measurement here can also allow a finite $\chi$-mismatch between the ancilla and the bosonic mode, as dephasing errors commute with the measurements. 

\subsection{X-basis measurement}
The $X$-basis measurement amounts to distinguishing the phase of the coherent states modulo $\pi$. We achieve this by interfering the mode $a_i$ with another mode $b_i$ in a coherent state $\ket{\alpha}$ through a $50:50$ beam splitter and measuring if the two output modes $a_o, b_0$ have less than $s$ photons. We obtain a logical $-$ if both modes have more than $s$ photons and a logical $+$ otherwise, i.e. we implement the following POVMs:
\begin{equation}
\begin{aligned}
    M_- & = [I_{a_o} - \sum_{i=0}^{s}(|i\rangle_{a_o}\langle i|)]\otimes [I_{b_o} - \sum_{i=0}^{s}(|i\rangle_{b_o}\langle i|)]\\
    & \approx (\hat{I}_{a_i} - \sum_{i=0}^{s}(|\alpha, i\rangle_{a_i} \langle \alpha, i|))(I_{a_i} - \sum_{i=0}^{s}(|-\alpha, i\rangle_{a_i} \langle -\alpha, i|)), \\
    M_+ &= I - M_-& \\
    & \approx \sum_{i=0}^{s}(|\alpha, i\rangle_{a_i} \langle \alpha, i|)) + \sum_{i=0}^{s}(|- \alpha, i\rangle_{a_i} \langle - \alpha, i|)),
\end{aligned}
\end{equation}
where each subscript labels the mode that a state or an operator belongs to.


Measuring if one mode has less than $s$ photons can be realized by dispersively coupling it to a qubit, driving the qubit from $\ket{g}$ to $\ket{e}$ conditioned on the mode having less than $s$ photons, and measuring the qubit in the $\ket{g}, \ket{e}$ basis. In the interaction picture associated with the dispersive coupling, the Hamiltonian reads
\begin{equation}
    \tilde{H}_{AC} = \Omega \left(\ket{e}\bra{g}\otimes P_{[s]} + h.c.\right).
\end{equation}
Recall that $P_{[s]} := \sum_{i=0}^s \ket{i}\bra{i}$.
In the absence of errors, the $0$-th order conditional operations are 
\begin{equation}
    \begin{aligned}
        \langle\langle e| \mathcal{G}^{[0]}(T) |g\rangle\rangle & = P_{[s]} \bullet P_{[s]} + O(p),\\
        \langle\langle g| \mathcal{G}^{[0]}(T) |g\rangle\rangle & = (I - P_{[s]}) \bullet (I - P_{[s]}) + O(p).\\
    \end{aligned}
    \label{eq:photon_number_measurement}
\end{equation}
A single fault will affect the measurement outcome or cause a bosonic error diagonal in the Fock basis. The former can be tolerated by repeating the above measurement three times and performing majority voting, while the latter simply commutes with the measurements and does not affect the (later) measurement outcome. 


To check this $X$-basis measurement scheme is $1$-FT according to Def.~\ref{def:FT_measurement}, we also need to verify that the measurement outcome is correct for any input error $(k,\theta) \leq (1, \Delta\chi T/2)$. First, a single-photon loss does not affect the measurement outcome since $a$ does not change the phase of any coherent states. Second, a small phase rotation rotates $\ket{\alpha}$ to $\ket{\alpha e^{i\theta}}$. Similar to the argument for the $X$-axis rotation, the $X$-basis measurement outcome is correct as long as the POVM $M_+$ captures $\ket{\pm \alpha e^{i\theta}}$ but not $\ket{\pm i\alpha e^{i\theta}}$.

\subsection{Error correction}
\begin{figure}
    \centering
    \includegraphics[width=0.5\textwidth]{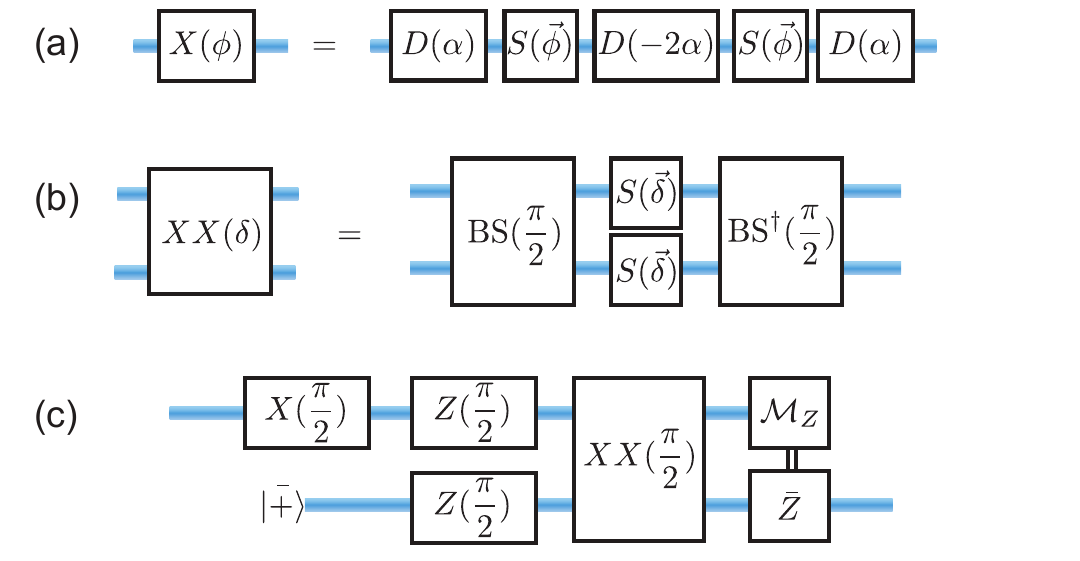}
    \caption{Illustration of the $X$-axis rotation (a), $XX$ rotation (b), and teleportation-based EC (c) in the level-$1$ gadgets $\mc{S}$ for the four-legged cat.
    }
    \label{fig:level_1_gadgets}
\end{figure}
To correct both loss and dephasing errors, i.e. data errors with $(k >0,  |\theta| >0)$, we employ a Knill-type QEC~\cite{knill2005scalable} using a teleportation circuit shown in Fig.~\ref{fig:level_1_gadgets}(c). 
A fresh ancilla bosonic mode $b$ is initialized in $\ket{+}$ state and gets entangled with the data mode $a$ via a $XX$ rotation along with singe-mode rotations. The data mode $a$ is then measured in the $Z$ basis, where the measurement outcome is used to apply a feedback $Z$ operation on the $b$ mode. 
All the gadgets here are $1$-FT using previous constructions. Moreover, they are error-transparent to any input error on the $a$ mode smaller than $f(1) = (1, \Delta \chi T/2)$. Therefore, the input data error simply commutes through all the gates and does not propagate to the $b$ mode. Furthermore, the $1$-FT $Z$-basis measurement is correct for an error smaller than $f(1)$. Therefore, such an input error can be corrected by the EC gadget. 

To verify the $1$-FT EC conditions, we need to further show that a single fault during the teleportation circuit only leads to a correctable residual error of size at most $f(1)$ at the output of the $b$ mode. Since we are using $1$-FT gates, the output for the $a$ or $b$ mode (before the $Z$ measurement) has an error at most $f(1)$.

\begin{figure}
    \centering
    \includegraphics[width=0.5\textwidth]{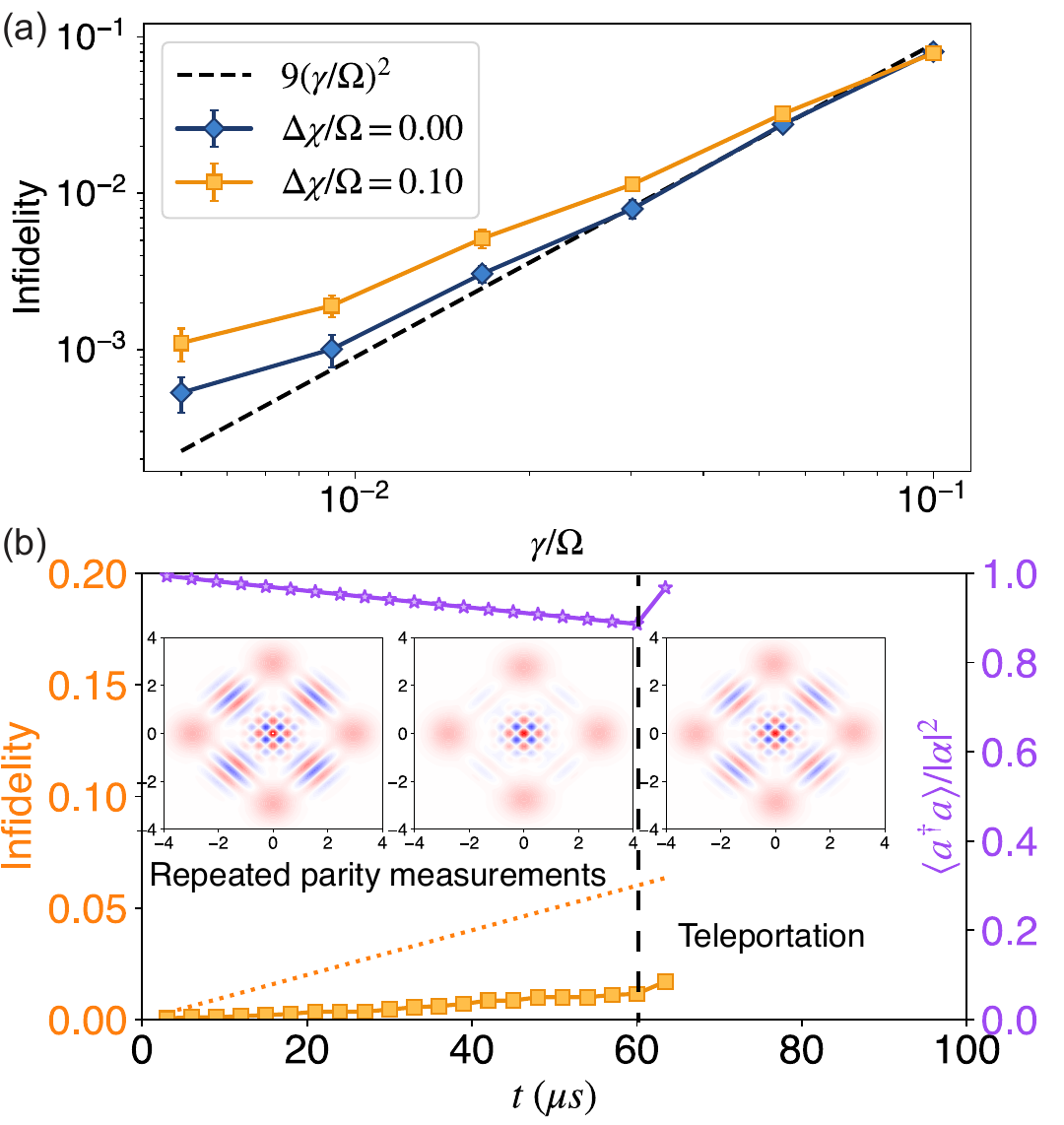}
    \caption{(a) Average infidelities of an error-correction gadget using teleportation in Fig.~\ref{fig:level_1_gadgets}(c) as a function of $\gamma/\Omega$ with perfect $\chi$ matching (blue line) or finite $\chi$ mismatches (orange line). Here, we use experimental parameters from Ref.~\cite{reinhold2020error} for the coherent interaction strengths $\chi_f = 2\pi\times 1$MHz, $\Omega = 0.3\chi_f$, $g_{\mr{BS}} = 2\chi_f$. 
    We consider single-photon loss, ancilla decay from $f$ to $e$, ancilla decay from $e$ to $g$, and ancilla dephasing $\mc{D}[\ket{e}\bra{e} + 2\ket{f}\bra{f}])$ with rates $\kappa$, $\gamma_{f\rightarrow e}$, $\gamma_{e\rightarrow g}$, and $\gamma_{\phi}$, respectively. We assume the ancilla error rates are much larger than the cavity loss rate and set $\gamma_{f\rightarrow e} = \gamma_{e\rightarrow g} = \gamma$, $\gamma_{\phi} = \gamma/4$, and $\kappa = \gamma/10$~\cite{reinhold2020error}.
    We choose $\alpha = 2.9$, which is a sweet spot for the four-legged cat that minimizes the back action of photon loss~\cite{li2017cat}. 
    (b) The accumulation of average infidelity and decay of mean photon number $\langle a^{\dagger}a\rangle$ for 40 rounds of repeated parity measurements (infidelities are shown for every two rounds) followed by teleportation. We use the same coherent parameters $\chi_f, \Omega$ and $g_{\mr{BS}}$ as in (a), a finite $\chi$ mismatch $\Delta \chi = \Omega/10$, and the experimental error rates from Ref.~\cite{reinhold2020error}: $\kappa = 2$KHz, $\gamma_{f\rightarrow e} = \gamma_{e \rightarrow g} = \gamma = 20$KHz and $\gamma_{\phi} = 5$KHz (with the same ratios between these error rates are in (a)). The teleportation pumps energy into the system and suppresses the random phase rotations caused by $\Delta\chi$. The three Wigner plots depict the density matrix at the input, before and after the teleportation respectively. 
    } 
    \label{fig:QEC_Infidelities}
\end{figure}


As shown in Fig.~\ref{fig:QEC_Infidelities}(a), we numerically evaluate the average infidelity of the teleportation gadget in Fig.~\ref{fig:level_1_gadgets}(c). In the absence of $\chi$ mismatch (see the blue curve),  we show that it has an error rate that scales as $(\kappa/\Omega)^2$, manifesting the fault tolerance of its composing gadgets, which cover the entire $\mc{S}$ other than the $X$-basis measurement. There is an error floor in the low $\kappa/\Omega$ regime, which is exponentially suppressed by $|\alpha|^2$, due to the finite-size imperfection of the $X$ rotation and the $XX$ rotation. In the presence of a finite $\chi$ mismatch, a rigorous second-order error suppression can no longer be attained due to the induced random phase rotations during the $X$- and $XX$-rotation gates. However, sufficient error suppression can still be achieved with a finite but small $\chi$-mismatch in practically relevant regimes (see the orange and gree curves). 

In practice, where photon loss is typically the predominant error source, repeated parity measurements that correct photon losses (see Alg.~\ref{alg:FT_QEC}) suffice for extending the lifetime of the cats. Such a robust memory that reaches the break-even point has been experimentally demonstrated~\cite{ofek2016extending}. 
However, only parity measurements are not enough to protect the cats during long computational tasks as the mean photon number would keep decaying (the parity measurement and gates in $\mc{S}$ are energy-preserving operations that commute with $a^{\dagger}a$) due to deterministic energy loss to the environment. We propose to solve this problem by inserting the teleportation gadget periodically in between certain rounds of parity measurements, which pumps energy into the system and restores the amplitude of the cats. Furthermore, the teleportation can suppress the accumulation of random phase rotations if, for example, there is some finite $\chi$-mismatch or small cavity dephasing errors $\kappa_{\phi}\mc{D}[a^{\dagger}a]$.
We demonstrate such effects numerically in Fig.~\ref{fig:QEC_Infidelities}(b). 

\section{Concatenated QEC and hardware-efficient fault-tolerance} \label{Sec:concatenatedQEC}
With the set of $1$-FT level-$1$ gadgets in $\mc{S}$, we can concatenate the level-$1$ logical qubits (four-legged cats) with a level-$2$ qubit code for arbitrary error suppression. We show such a concatenated architecture in Fig.~\ref{fig:concatenation_layout}. The basic elements for each level-$1$ qubit are simply a storage bosonic mode and a three-level ancilla that are dispersively coupled. The ancilla is used for the fault-tolerant operation of the bosonic qubit in each storage mode, including state preparation, photon-loss correction, gates, and measurements (see Sec.~\ref{sec:FT4legCat}).  In addition, a small number of extra bosonic modes shared by neighboring storage modes, which we refer to reservoir modes, are used to pump energy into the storage modes periodically via teleportation (see Fig.~\ref{fig:level_1_gadgets}(c)).   
\begin{figure}
    \centering
    \includegraphics[width=0.5\textwidth]{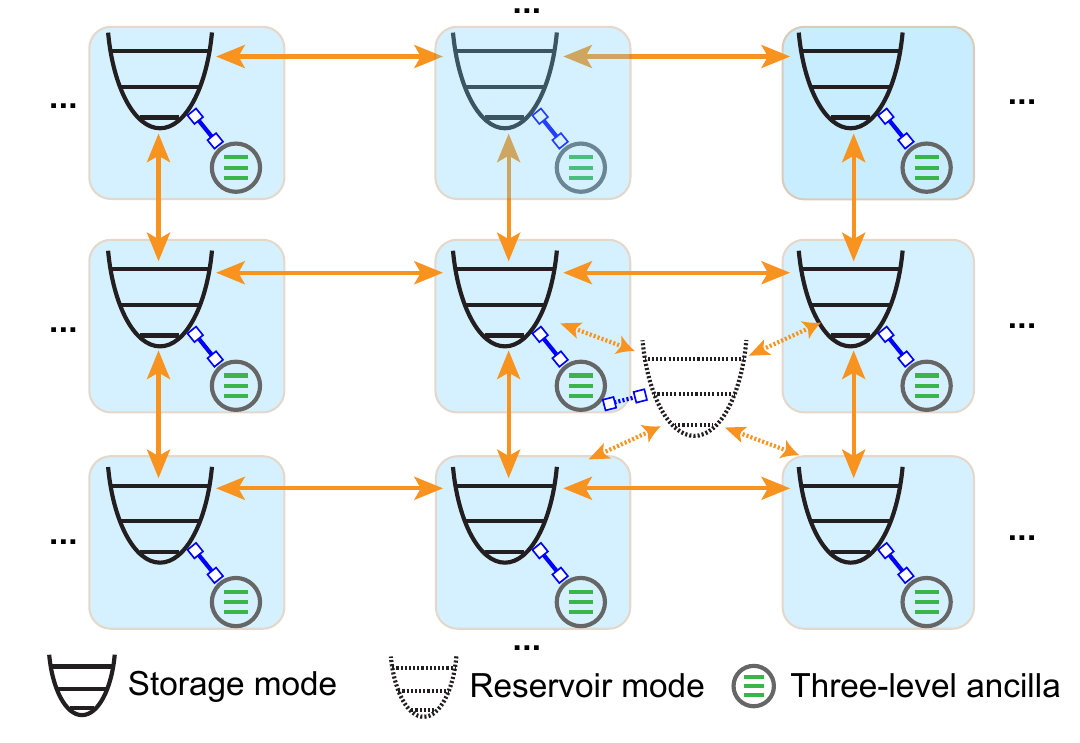}
    \caption{Hardware layout for concatenated 2D codes with four-legged cats. Each level-$1$ logical qubit (blue box) consists of a storage bosonic mode and a three-level ancilla, which are dispersively coupled. BS coupling between neighboring storage bosonic modes is required for the level-$2$ QEC. In addition, reservoir modes (with only one shown here as an example) shared between neighboring storage modes are used to pump energy into the system via teleportation (see Fig.~\ref{fig:level_1_gadgets}(c)). }
    \label{fig:concatenation_layout}
\end{figure}

The level-$2$ QEC requires certain couplings between level-$1$ qubits. Importantly, we can achieve this by introducing only BS coupling between nearest-neighbor storage bosonic modes (see Fig.~\ref{fig:concatenation_layout}) for 2D stabilizer codes. The level-$2$ syndrome-extraction circuits are made of non-destructively measurement of high-weight Pauli operators, featuring a sequence of two-qubit entangling gates such as the CNOT gate. In Fig.~\ref{fig:concatenation_circuit}(a), we show how one can get a level-$1$ CNOT gate using $1$-FT single-mode and two-mode rotations in $\mc{S}$. Although the complied circuit is long with a depth $6$, we note that one can usually reduce the depth per CNOT gate when considering the full stabilizer measurement circuits. As an example, we can measure a weight-$n$ $X$ Pauli operator using a circuit of depth $2n + 4$ (see Fig.~\ref{fig:concatenation_circuit}(b)). We leave the evaluation and optimization of the error rates of level-$1$ gates, e.g. the CNOT gate, as well as the threshold and resource overhead of concatenated codes to future work. Nevertheless, we remark that each CNOT gate (se Fig.~\ref{fig:concatenation_circuit}(a)) uses similar gadgets as those for teleportation (see Fig.~\ref{fig:level_1_gadgets}(c)), and each CNOT gate in a syndrome extraction depth (see Fig.~\ref{fig:concatenation_circuit}(b)) has a similar depth as the teleportation on average, we expect the CNOT gates have a similar error rate as the teleportation shown in Fig.~\ref{fig:QEC_Infidelities}(b). Using this rough estimates, a gate error rate below $10^{-2}$, which corresponds to the gate error threshold for the surface codes, is achievable using the current circuit-QED hardware.

To sum up, our construction of $\mc{S}$ in this work enables a practical, hardware-efficient architecture for fault-tolerant quantum computing, which features only one bosonic mode and one qutrit per level-$1$ logical qubit and only requires low-order physical interactions (dispersively coupling and beam-splitter interaction) that have been experimentally demonstrated. Furthermore, realizing high-fidelity level-$1$ gadgets with error rates below the threshold requirement of the level-$2$ codes is promising for near-term experimental demonstrations.
\begin{figure}
    \centering
    \includegraphics[width=0.5\textwidth]{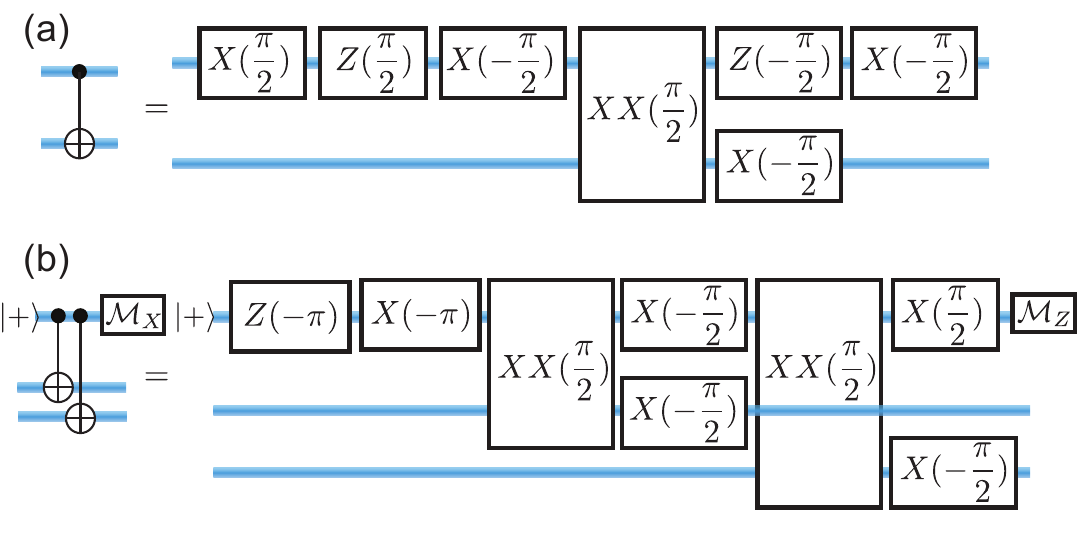}
    \caption{Compilation of level-$1$ CNOT (a) and a stabilizer $X^{\otimes 2}$ measurement circuit (b) using our constructed $1$-FT level-$1$ gadgets in $\mc{S}$. }
    \label{fig:concatenation_circuit}
\end{figure}


\section{Discussion \label{sec:discussion}}
The fault-tolerant gadgets $\mc{S}$ that we develop in Sec.~\ref{sec:FT4legCat} for the four-legged cat can be applied to other rotation-symmetric codes~\cite{grimsmo2020quantum}, whose codewords are invariant under a $N$-fold rotation $R = \exp[i (2\pi/N) a^{\dagger}a]$. Taking $N = 2$ for example, an arbitrary code with a two-fold rotation symmetry have codewords
\begin{equation}
    \begin{aligned}
    \left|+_{\Theta}\right\rangle & \approx \frac{1}{\sqrt{\mathcal{N}_+}} (I + e^{i\pi a^{\dagger}a})|\Theta\rangle, \\
    \left|-_{\Theta}\right\rangle & \approx \frac{1}{\sqrt{\mathcal{N}_-}} (e^{i \pi a^{\dagger}a/2} + e^{i 3\pi a^{\dagger}a/2})|\Theta\rangle, \\
    \end{aligned}
\end{equation}
where $\mc{N}_{\pm}$ are normalization constants, and the approximation holds when the base state $\ket{\Theta}$ is localized in phase space, i.e. $\bra{\Theta}e^{i \pi a^{\dagger}a/2}\ket{\Theta} \approx 0$. The fault-tolerant gadgets in $\mc{S}$ can be applied to such an arbitrary rotation-symmetric code with a localized base state $\ket{\Theta}$, except that for the $X$-basis state preparation in Alg.~\ref{alg:FT_X_preparation}, we need to replace the initial state with $\ket{\Theta}$ in the first step. In particular, the $X$ rotation and $XX$ rotation still work since they are based on the phase-space locality of the base state.

\setlength{\tabcolsep}{10pt} 
\begin{table*}
\centering
\begin{tabular}{ccc}
\hline \hline 
\makecell{Gadgets} & \makecell{Prior schemes} & \makecell{Our scheme}\\
\hline 
\makecell{Error correction} & \makecell{PI parity measurement~\cite{rosenblum2018fault};\\
Two one-bit teleportation \\
with two ancillary bosonic modes~\cite{grimsmo2020quantum}.\\
Engineered dissipation~\cite{xu2023autonomous, gertler2021protecting}.} & \makecell{GPI parity measurement \\
+ one-bit teleportation with a shared ancillary mode.}   \\
\hline
\makecell{$Z$-type gates}
 & \makecell{PI SNAP gate~\cite{reinhold2020error, ma2020path};\\
 Self-Kerr $(a^{\dagger}a)^2$ Hamiltonian~\cite{grimsmo2020quantum}.} & GPI SNAP gate\\
 \hline
\makecell{$X$-type gates} & \makecell{Teleported Hadamard gate \\ with an ancillary bosonic mode~\cite{grimsmo2020quantum}.} & \makecell{$X$-axis rotation \\
using cavity displacements and SNAP gates} \\
\hline 
\makecell{Entangling gate} & \makecell{CZ gate using cross-Kerr $a^{\dagger}a\otimes b^{\dagger}b$~\cite{grimsmo2020quantum}}. & $XX$ rotation using beam-splitter + SNAP gates.\\
\hline
$X$-basis measurement & Phase measurement~\cite{grimsmo2020quantum}. & Beam splitter + SNAP gates.\\
\hline
\end{tabular}
\caption{Comparison of different constructions of fault-tolerant gadgets for rotation-symmetrical codes that can correct photon losses. We denote $Z$-type gates as those that preserve the photon number (alternatively, those that add photon-number dependent phases), and $X$-type gates as those that do not preserve the photon number. }
\label{tab:scheme_comparison}
\end{table*}

In Tab.~\ref{tab:scheme_comparison}, we compare our construction of fault-tolerant gadgets for rotation-symmetrical codes that can correct photon losses with those in the literature. In particular, compared to the gadgets in Ref.~\cite{grimsmo2020quantum} using bosonic ancillae, our gadgets using qutrit ancillae avoid the demand of nonlinear interaction between bosonic modes and the phase measurement, which are both challenging to engineer in practice. 

\begin{acknowledgments}
We thank Wenlong Ma and Takahiro Tsunoda for helpful discussions. The authors acknowledge support from the ARO (W911NF-23-1-0077), ARO MURI (W911NF-21-1-0325), AFOSR MURI (FA9550-19-1-0399, FA9550-21-1-0209), NSF (OMA-1936118, ERC-1941583, OMA-2137642), NTT Research, Packard Foundation (2020-71479), and the Marshall and Arlene Bennett Family Research Program. This material is based upon work supported by the U.S. Department of Energy, Office of Science, National Quantum Information Science Research Centers. 
The authors are also grateful for the support of the University of Chicago Research Computing Center for assistance with the numerical simulations carried out in this paper.
\end{acknowledgments}
\appendix
\section{Algebric conditions for PI and GPI}
\label{appendix:algebric_conditions}
In this section, we provide algebraic conditions for PI gates (Def.~\ref{def:PI_gate}, Def.~\ref{def:finite-order_PI_gate}) and GPI gates (special case of Def.~\ref{def:GPI_operation} when the target operation is a unitary). 

Recall that we are considering a Markovian evolution for the joint system of ancilla $A$ and bosonic mode $C$, described by the Lindblad master equation in Eq.~\eqref{eq:master_equation_ancillla_errors_only}, with a joint Hamiltonian $H_{AC}(t)$ and a set of ancilla errors $\{J_j\}$.

We first provide definitions and properties of a set of structured algebras that we will use.

Let $\mathcal{B}_A = \{\ket{m}_A\}_{m \in [d_A]}$ be an orthonomal basis for a $d_A$-dimensional ancilla, and $\mathcal{B}_C = \{\ket{n}_C\}_{n \in [\infty]}$ be an orthonomal basis for an infinite-dimensional bosonic mode. Let $\mb{M}_A$ ($\mb{M}_C$) be the ordinary matrix algebra for the ancilla (bosonic mode). 
$\mb{M}_A$ ($\mb{M}_C$) is a vector space over $\mathbb{C}$ with a basis $\{\ket{m}_A\bra{n}\}_{m,n \in [d_A]}$ ($\{\ket{m}_C\bra{n}\}_{m,n \in [\infty]}$). In addition, $\mb{M}_A$ (similarly for $\mb{M}_C$) is equipped with a multiplication operation 
\begin{equation}
    \ket{a}_A\bra{b} \ket{c}_A\bra{d} = \delta_{b,c} \ket{a}_A\bra{d},
\end{equation}
for any $a,b,c,d \in [d_A].$
For any algebra $\mb{M}$, we denote $\mb{M} = \langle \mb{S}\rangle$ for a set $\mb{S}$ if any element $a$ in $\mb{M}$ can be written as $a = \sum_j c_j \alpha_j$, where $c_j \in \mathbb{C}$ and $\alpha_j$ is a product of some elements in $\mb{S}$. In other words, $\mb{M}$ is generated by $\mb{S}$. Let $\mb{M}_{AC} = \mb{M}_A \otimes \mb{M}_C$ be the matrix algebra for the joint system. 

We define the reduction of an algebra on the joint system to an algebra only on the ancilla as a surjective map from $\mb{M}_{AC}$ to $\mb{M}_A$:
\begin{definition}[Reduction of a joint algebra]
Given any algebra $\mb{H} \subseteq \mb{M}_{AC}$ on the joint system and an ancilla basis $\mc{B}_A$ we define the reduction of $\mb{a}$ on $\mc{B}_A$ as:
\begin{equation}
\begin{aligned}
    \mb{H}|_{\mc{B}_A} := \langle \{\ket{m}\bra{n} \big| \ket{m}, \ket{n} \in \mc{B}_A; \\
    \exists h \in \mb{H}, \bra{m}h\ket{n} \neq 0\} \rangle.
\end{aligned}
\end{equation}
\end{definition}

Next, we define a family of subalgebras of $\mb{M}_{AC}$ that satisfy a special property:
\begin{definition}[PI matrix algebra. Definition 1 of Ref.~\cite{ma_algebraic_2022}]
Let $\mathcal{B}_A$ be some orthonormal basis for $A$. We say that a subalgebra $\mb{P}$ of $\bf{M}_{AC}$ is a PI algebra associated with $\mathcal{B}_A$ if it satisfies:
\begin{enumerate}
    \item  $\mathbf{P} = \langle\{\ket{m}\bra{n}\otimes U_{mn}\}\rangle$ for some set of $(m,n) \in [d_A]\times[d_A]$ and $(m,n)$-dependent unitaries $U_{mn} \in \mb{M}_C$.
    \item $\mathbf{P}$ is isomorphic to its reduction on $\mc{B}_C$ via the map $\ket{m}\bra{n}\otimes U_{mn} \rightarrow \ket{m}\bra{n}$.
\end{enumerate}
\label{def:PI_algebra}
\end{definition}
Note that the second condition posts three requirements on the unitaries $U_{mn}$:
\begin{enumerate}
    \item $U_{ma}U_{bn} = \delta_{a,b} U_{mn}$.
    \item $U_{mn} = U_{nm}^{\dagger}$.
    \item $U_{mm} = I$.
\end{enumerate}

These requirements lead to the following properties of operators in a PI algebra:
\begin{proposition}[Property of operators lying in a PI algebras]
Let $\mathbf{P} = \langle\{\ket{m}\bra{n}\otimes U_{mn}\}\rangle$ be some PI algebra associated with an ancilla basis $\mathcal{B}_A$, and let $\mb{P}|_{\mc{B}_A}$ be its reduction on $\mc{B}_A$. For any operator $O \in \mb{P}$ and $\ket{r}, \ket{i} \in \mathcal{B}_A$, we have $\bra{r}O\ket{i} \propto U_{ri}$, where $U_{ri} := I$ if $\ket{r}\bra{i} \notin \mb{P}|_{\mc{B}_A}$.
\label{prop:PI_implication}
\end{proposition}
\begin{proof}
    If $O \in \mb{P}$, we can write $O$ as $O = \sum_{m,n} o_{mn} \ket{m}\bra{n}\otimes U_{mn}$ for some $o_{mn}\in\mc{C}$. If $\ket{r}\bra{i} \in \mb{P}|_{\mc{B}_A}$, we have $\bra{r}O\ket{i} = o_{ri} U_{ri} \propto U_{ri}$; If  $\ket{r}\bra{i} \notin \mb{P}|_{\mc{B}_A}$, we have $\bra{r}O\ket{i} = 0\times I$. 
\end{proof}
Note that Prop.~\ref{prop:PI_implication} also implies that for any operator $O = \prod_i O_i$ that is a product of operators lying in a PI algebra $\mb{P} = \langle \{\ket{m}\bra{n}\otimes U_{mn}\}\rangle$, we have $\bra{r}O\ket{i} \propto U_{ri}$. 

\subsection{PI conditions}
Ref.~\cite{ma_algebraic_2022} provides a simple algebraic condition for PI gates using PI algebras:



\begin{proposition}[Algebraic condition for PI gates. Theorem 1 of Ref.~\cite{ma_algebraic_2022}. ]
An ancilla-assisted gate is PI (see Def.~\ref{def:PI_gate})  in an ancilla basis $\mathcal{B}_A$ if the Hamiltonian and all the Lindblad jump operators are all in some PI algebra associated with $\mathcal{B}_A$. 
\label{prop:PI_algebraic_conditions}
\end{proposition}

Note that the condition in Prop.~\ref{prop:PI_algebraic_conditions} guarantees PI up to infinite order. In the following, we generalize it for finite-order PI gates.

First, the efffective Hamiltonian $H_{\mr{eff}} = H_{AC}(t) - \frac{i}{2}\sum_j J_j^{\dagger}J_j$ needs to be in some PI algebra, i.e. $H_{\mr{eff}}(t) = \sum_{m,n = 1}^{d_A} \xi_{mn}(t) \ket{m}\bra{n} \otimes U_{mn}$ for some $\xi_{mn}(t) \in \mathbb{C}$ and unitaries $U_{mn}$. 

Next, we define a $n$-th order path algerba $\mathbf{p}^{[n]}$ containing all the possible paths contained in the noise expansion of the system dynamics up to $n$-th order, and an associated $n$-th order reachable state set $\mathbf{S}_A^{[n]}$ containing all ancilla states reachable via the $n$-th order paths in  $\mathbf{p}^{[n]}$ when starting from $\ket{i}$, and a $n$-th order error set $\mb{E}^{[n]} \subseteq \{J_j\}$ containing all possible errors that can act nontrivially on $\mathbf{S}_A^{[n-1]}$:
\begin{definition}[Finite-order path algebras, reachable states, and error sets]
Given a Hamiltonian $H_{AC}(t)$ that lies in some PI algebra associated with an ancilla basis $\mc{B}_A$, a set of errors $\{J_j\}$, and an initial ancilla state $\ket{i}$, we define the zeroth-order path algebra $\mathbf{p}^{[0]}$ as an algebra that contains all the paths in the effective Hamiltonian $H_{\mr{eff}}(t) := H_{AC}(t) - \frac{i}{2}\sum_j J_j^{\dagger}J_j = \sum_{m,n = 1}^{d_A} \xi_{mn}(t) \ket{m}\bra{n} \otimes U_{mn}$:
\begin{equation}
    \mb{p}^{[0]} := \langle \{\ket{m}\bra{n}\otimes U_{mn}\mid \exists t \in [0, T], \xi_{mn}(t) \neq 0\}\rangle,
\end{equation}
Let $\mb{p}^{[0]}|_{\mc{B}_A}$ be the reduction of $\mb{p}^{[0]}$ on $\mc{B}_A$.
We define a zeroth order reachable state set including all states reachable via the zeroth order paths when starting from $\ket{i}$: 
\begin{equation}
    \mb{S}_{A}^{[0]} := \{\ket{m} \in \mc{B}_A \mid \ket{m} \bra{i} \in \mb{p}^{[0]}|_{\mc{B}_A}\},
\end{equation}
and we define a zero-th order error set $\mb{E}^{[0]} := \emptyset$.

Then, we define a $n \geq 1$-th order path algebra $\mathbf{p}_{AC}^{[n]}$ and a $n$-th order reachable state set inductively:
\begin{equation}
\begin{aligned}
\mb{E}^{[n]}& = \mb{E}^{[n]} \cup \{J_j \,\big|\, \exists \ket{m} \in \mathbf{S}_A^{[n-1]}, J_j \ket{m} \neq 0\},\\
    \mathbf{p}^{[n]} & = \langle\mathbf{p}^{[n-1]}\cup \mb{E}^{[n]}]\rangle,\\
    \mathbf{S}_A^{[n]} & = \{\ket{m} \,\big| \,\ket{m} \bra{i} \in \mb{p}^{[n]}|_{\mc{B}_A}\}.
\end{aligned}
\end{equation}
\end{definition}

\begin{proposition}[Algebraic conditions for finite-order PI gates]
Given an ancilla-assisted gate $\mc{G}(T)$ generated by a Hamiltonian $H_{AC}(t)$ and jump errors $\{J_j(t)\}$. $\mc{G}(T)$ is $n$-GPI in an ancilla basis $\mc{B}_A$ for an initial ancilla state $\ket{i}$ if $H_{AC}(t) \cup \{J_j^{\dagger}J_j\}\cup \mb{E}^{[n]}$ are in some PI algebra, where $\mathbf{E}^{[n]}$ is the and $n$-th order error set  constructured from $(H_{AC}(t), \{J_j\}, \mc{B}_A, \ket{i})$.  
\label{prop:finite_PI_algebraic_condition}
\end{proposition}
\begin{proof}
    Let $H_{AC}(t)\cup \{J_j^{\dagger}J_j\}\cup \mb{E}^{[n]}$ be in some PI algebra  $\mathbf{P} = \langle \{\ket{m}\bra{n}\otimes U_{mn}\} \rangle$. 
  The effective Hamiltonian $H_{\mathrm{eff}}(t) = H_{AC}(t) - \frac{i}{2}J_j^{\dagger}J_j \in \mathbf{P}$.
    Furthermore, the non-jump propagator $W(t_2, t_1) := \mc{T} \exp[-i \int_{t = t_1}^{t_2} dt H_{\mathrm{eff}}(t)]$ is also in $\mb{P}$.
    Let $\mathcal{I}_k := \{j \mid J_j \in \mathbf{E}^{[k]}\}$ be the indices of the $k$-th order error set. Now, consider a $k$-th order Dyson expansion of $\mathcal{G}(T)$ (see Eq.~\eqref{eq:dyson_expansion_Kraus_form}): 
    \begin{equation}
    \begin{aligned}
        \langle\langle r| \mathcal{G}_k(T)|i\rangle\rangle = & \left[\int_{t_h = 0}^T d t_h\right]_{h \in [k]} \left[\sum_{j_h \in \mathcal{I}_k}\right]_{h \in [k]}
        \\
        & G_{ri}^k(\{t_h, j_h\}_{h\in [k]}) \bullet G^{\dagger k}_{ri}(\{t_h, j_h\}_{h \in [k]}),
        \label{eq:kraus_expansion}
    \end{aligned}
    \end{equation}
    where 
    \begin{equation}
    \begin{aligned}
        G^k_{ri}(\{t_h, j_h\}_{h\in[k]}) = & \bra{r} \mathcal{T} \Bigg\{
        W(T, t_k) \\
        & \left.\times \prod_{h=1}^k K_{j_h}(t_h) W(t_h, t_{h-1})\right\}\ket{i},
        \label{eq:traj_oper}
    \end{aligned}
    \end{equation} 
    with $t_0 = 0$.  Since all the operators in Eq.~\eqref{eq:traj_oper} are in $\mathbf{P}^{\prime}$, we have $G^k_{ri}(\{t_h, j_h\}_{h \in [k]}) \propto U_{ri}$ according to Prop.~\ref{prop:PI_implication}. As such, $\langle\langle r| \mathcal{G}_k(T)|i\rangle\rangle \propto U_{ri}\bullet U_{ri}$. Therefore, $\bbra{r}\mc{G}^{[k]}(T)\kket{i} \propto U_{ri}\bullet U_{ri}$ for any $k \leq n$ and the $n$-PI condition in Eq.~\eqref{eq:def_finite_PI} is satisfied.
    \label{prop:finite-order_PI_condition}
\end{proof}

\subsection{GPI conditions}
Here, we further relax the condition for finite-order PI in Prop.~\ref{prop:finite_PI_algebraic_condition} and provide algebraic conditions for finite-order GPI gates. 

\begin{proposition}[Algebraic conditions for finite-order GPI gates] 
Given a bosonic code with a code projection $P_c$ and an ancilla-assisted gate generated by a Hamiltonian $H_{AC}(t)$ and jump errors $\{J_j(t)\}$.
$\mc{G}(T)$ is $n$-GPI in an ancilla basis $\mathcal{B}_A$ for an initial ancilla state $\ket{i}$ if
\begin{enumerate}
    \item There exists some PI algebra $\mathbf{P} = \langle \{\ket{m}\bra{n}\otimes U_{mn}\} \rangle$ such that $H_{AC}(t) \in \mathbf{P}$, and any error $J_j(t) \in \mathbf{E}^{[n]}$, where $\mb{E}^{[n]}$ is the $n$-th order error set constructured from $(H_{AC}(t), \{J_j\}, \mc{B}_A, \ket{i})$,
    is in the form $J_j(t) = \sum_{m,n} \ket{m}\bra{n} \otimes R_{mn}^j(t) U_{mn}$, where $R^j_{mn}(t)$ are unitaries. 

    \item Let $\mathbf{\xi} := \{R_{mn}^j(t) \mid J_j \in \mathbf{E}^{[n]}; m,n \in [d_A]; t \in [0, T] \}$. Any error $E \in \mathbf{\xi}$ satisfies
    \begin{equation}
        [E, U_{mn}] = 0.
    \end{equation}
    \item Let $\epsilon := \{\mathcal{T} \prod_{i = 1}^n E_i(t_i) \}_{E_i(t_i) \in \xi \cup I}$. Errors in $\epsilon$ satisfy the Knill-Laflamme condition with respect to $P_c$.  
\end{enumerate}
\label{prop:GPI_algebraic_condition}
\end{proposition}
\begin{proof}
    We follow the same proof as that for Prop.~\ref{prop:finite_PI_algebraic_condition}. Now, each Kraus operator $G^k_{ri}(\{t_h, j_h\})$ of $\langle\langle r|\mathcal{G}_k(T)|i\rangle\rangle$ (see Eq.~\eqref{eq:kraus_expansion} and Eq.~\eqref{eq:traj_oper}) reads
    \begin{equation}
    \begin{aligned}
        G^k_{ri}(\{t_h, j_h\})
        & = U_{ri} \sum_{m_h, n_h} c_{m_h, n_h}\left[ \mathcal{T} \prod_{h=1}^k R^{j_h}_{m_h n_h}(t_h) \right],
    \end{aligned}  
    \end{equation}
    for some $c_{m_h, n_h} \in \mathbb{C}$. Here, we have replaced the form of $J_j$ in the first condition of Prop.~\ref{prop:GPI_algebraic_condition} into Eq.~\eqref{eq:traj_oper}. The operators $R^{j_h}_{m_h n_h}$ cumulate to the front as if they were transparent to the unitaries due to the second condition in Prop.~\ref{prop:GPI_algebraic_condition}. Then, the Kraus operators of $\langle\langle r| \mathcal{G}_k(T) |i\rangle\rangle$, given by the union of all $G^k_{ri}(\{t_h, j_h\})$, are all linear combinations of errors in $\{U_{ri} \mathcal{T} \prod_{i = 1}^k E_i(t_i) \}_{E_i(t_i) \in \xi}$, where $\xi$ is defined in the second condition of Prop.~\ref{prop:GPI_algebraic_condition}. Finally, the Kraus operators of $\langle\langle r| \mathcal{G}^{[n]}(T) |i\rangle\rangle$ are all linear combinations of errors in $\epsilon$, up to a same unitary $U_{ri}$. Therefore, the errors are correctable if $\epsilon$ satisfies the KL condition Prop.~\ref{prop:GPI_algebraic_condition}.
\end{proof}

As an example, we consider the GPI SNAP gate in Sec.~\ref{sec:GPI_SNAP} using a $\chi$-mismatched three-level transmon. 
In the presence of the ancilla relaxations, one can show that this gate is $1$-GPI by checking the conditions in Proposition~\ref{prop:GPI_algebraic_condition}. 
In the interaction picture, the first-order error set $\mb{E}^{[1]} = \{\ket{e}\bra{f}\otimes e^{-i\Delta\chi t a^{\dagger}a }\}$. We can find a PI algebra such that the first condition of Prop.~\ref{prop:GPI_algebraic_condition} is satisfied:
\begin{equation}
    \mb{P} = \langle\{\ket{f}\bra{g}\otimes S, \ket{g}\bra{f}\otimes S^{\dagger}, \ket{e}\bra{f}\otimes I\}\rangle.
\end{equation}
Here, there is only a single $J_j$ with $R_{ef}(t) = e^{-i\Delta\chi t a^{\dagger}a}$. Therefore, $\xi = \{R_{ef}(t)\}_{t \in [0, T]}$ and $\epsilon = \xi\cup I$.
By choosing $S$ as a logical operator for the four-legged cat and noticing that $[R_{ef}(t), S] = [R_{ef}(t), S^{\dagger}] = 0$, the second condition of Proposition~\ref{prop:GPI_algebraic_condition} is also satisfied.
Finally, $\epsilon = \{R_{ef}(t)\}_{t \in [0, T]} \cup I$ satisfies the KL condition w.r.t. $P_c$ of the four-legged cat as long as $\Delta \chi T < \pi/2$ (see Sec.~\ref{sec:bosonic_code}).


\section{Error transparent/closure control for bosonic errors\label{appendix:error_transparency}}
In this section, we review the error-transparent~\cite{ma_error-transparent_2020} and error-closure~\cite{tsunoda2023error} quantum control techniques that enable fault tolerance against central system (bosonic) errors. For bosonic errors, we neglect their contribution to the no-jump evolution (the no-jump propagator is purely generated by the Hamiltonian) in the jump expansion of a quantum channel (see Eq.~\eqref{eq:dyson_expansion}). Such an approximation can be justified when considering the photon loss, whose back-action associated with $a^{\dagger}a$ is correctable in the large-$\alpha$ regime for cat codes (see Sec.~\ref{sec:bosonic_code}).

We consider a unitary $U(t)$ generated by a Hamiltonian $H(t)$ that acts only on the bosonic mode. Consider a bosonic code with a code projection $P_c$ and a bosonic error $E$. The error-transparent control aims to engineer $H(t)$ such that the dynammics is transparent to $E$, i.e. errors that occur during the gate are equivalent to those that occur after the unitary:
\begin{equation}
   U(T,t_p) E \cdots  E U(t_2, t_1) E U(t_1, 0) P_c \propto E^p U(T,0) P_c,
   \label{eq:error-transparency_condition1}
\end{equation}
for any $T > t_p > \cdots > t_1 > 0$ and $p \geq 1$. We say that the unitary is $k$-th order error-transparent to $E$ if Eq.~\eqref{eq:error-transparency_condition1} is satisfied for $p \leq k$.

When using $U(T,0)$ as a logical gate for the bosonic code, we typically require $H(t)$ (and thereby, $U(t)$) to be block-diagonal, i.e. $(I - P_c) H(t) P_c = 0$. In this case, Eq.~\eqref{eq:error-transparency_condition1} is equivalent to
\begin{equation}
    U(T, t) E U^{\dagger}(T,t) P_j \propto E P_j,
    \label{eq:error-transparency_condition2}
\end{equation}
for any $T > t >0$ and $j \leq k - 1$, where $P_j := E^j P_c$. 
Obviously, Eq.~\eqref{eq:error-transparency_condition2} is satisfied if $E$ commutes with $H(t)$ when acting on the error spaces up to $(k-1)$-th order, i.e.
\begin{equation}
    [E, H(t)]P_j = 0,
    \label{eq:error-transparency_condition3}
\end{equation}
for $j \leq k - 1$.
We note that Eq.~\eqref{eq:error-transparency_condition3} is only a sufficient condition for the error-transparency definition in Eq.~\eqref{eq:error-transparency_condition2}. For instance, Eq.~\eqref{eq:error-transparency_condition2} is also satisfied if $[E, H(t)]P_j \propto E P_j$. 

In this work, we are interested in ancilla-assisted gates. Similar to the PI/GPI, in the case where we initialize the ancilla in $\ket{i}$ and projectively measure it in a basis $\mathcal{B}_A = \{\ket{m}_A\}$, we only care about the conditional bosonic channels given a measurement outcome $m$. As such, we consider the following conditional error transparency, which is easier to achieve than unconditional error-transparency presented above:
\begin{definition}[Conditional error transparency]
Given a bosonic code with a code projection $P_c$, an initial ancilla state $\ket{i}$ and an ancilla orthonormal basis $\mathcal{B}_A = \{\ket{m}_A\}$, we say that an ancilla assisted unitary $U(t)$ is $(P_c, \ket{i}, \mathcal{B}_A)$-error-transparent to a bosonic error $E$ up to $k$-th order if for any $\ket{r} \in \mathcal{B}_A$ and $p \leq k$:
\begin{equation}
    \bra{r} U(T,t_p) E \cdots E U(t_1, 0) P_c \propto E^p \bra{r} U(T,0) \ket{i} P_c,
   \label{eq:conditional_error-transparency}
\end{equation} 
\label{def:conditional_error_transparency}
\end{definition}

Ref.~\cite{tsunoda2023error} generalizes the error transparency condition to the so-called error closure condition. In the case of a single error $E$ and a static Hamiltonian $H_0$, they first construct a vector space $\epsilon$ with a basis $\{E, [H_0, E]\}$ over $\mathbb{C}$, and the error-closure condition is satisfied if for any $e \in \mb{\epsilon}$:
\begin{enumerate}
    \item $[H_0, e] \in \epsilon$.
    \item Errors in $\epsilon$ are correctable (satisfying the KL condition) with respect to $P_c$. 
\end{enumerate}
Such a condition guarantees that each first-order error trajectory gives:
\begin{equation}
    e^{iH_0 (T-t)}E e^{iH_0 t} = e^{iH_0 (T-t)} E e^{-iH_0 (T-t)} e^{iH_0 t} = E^{\prime} e^{iH_0 t},
\end{equation}
where $E^{\prime} := e^{iH_0 (T-t)} E e^{-iH_0 (T-t)} \in \epsilon$ using the first condition. Then the desired unitary is implemented up to a correctable error $E^{\prime}$ according to the second condition. The error closure condition generalizes the error transparency condition as it allows errors to propagate to correctable errors, rather than rigorously commuting through the unitary, in a similar spirit as our generalization from PI to GPI.  

\subsection{$Z$-axis rotation \label{sec:ET_SNAP}}
Recall that a $Z$-axis rotation is implemented by a GPI SNAP gate (see Sec.~\ref{sec:GPI_SNAP}). In the interaction picture associated with the base Hamiltonian $H_0 = - (\chi_f \ket{f}\bra{f} + \chi_e \ket{e}\bra{e})\otimes a^{\dagger}a$, the Hamiltonian is static $\tilde{H} = \Omega \left[ \ket{f}\bra{g} \otimes S(\vec{\phi}) + h.c. \right]$, and the photon loss error reads $\tilde{a}(t) = e^{i(\chi_f \ket{f}\bra{f} + \chi_e \ket{e}\bra{e})t}\otimes a$. Note that $[\tilde{a}(t), \tilde{H}]I\otimes P_c \neq 0$ and the unconditional error transparency does not hold. Fortunately, we now show that the conditional error transparency in Eq.~\eqref{eq:conditional_error-transparency} holds up to a single-photon loss if we choose $S(\vec{\phi})$ appropriately. For $p = 1$, the L.H.S. of Eq.~\eqref{eq:conditional_error-transparency} reads
\begin{equation}
\begin{aligned}
    & \bra{r} \tilde{U}(T, t) \tilde{a}(t) U(t,0)\ket{g} P_c \\
    &= \sum_{m \in \{f,g\}} e^{i\chi_m t}
 \tilde{U}_{rm}(T,t) a \tilde{U}_{m g}(t,0) P_c.
\end{aligned}
\label{eq:conditional_ET_eval}
\end{equation}
Recall that $\tilde{U}(t_2, t_1)$ is in a PI algebra (see Appendix~\ref{appendix:algebric_conditions}) $\mathbf{P} = \langle\{\ket{f}\bra{g}\otimes S, \ket{g}\bra{f}\otimes S^{\dagger}, \ket{g}\bra{g}\otimes I, \ket{e}\bra{e}\otimes I, \ket{f}\bra{f}\otimes I\}\rangle$. Therefore, $\tilde{U}_{fg} \propto S(\vec{\phi})$, $\tilde{U}_{gf} \propto S^{\dagger}(\vec{\phi})$, and $\tilde{U}_{gg}, \tilde{U}_{ff}  \propto I$. Choosing $S(\vec{\phi})$ as a logical gate, we have $\tilde{U}_{m g}(t,0) P_c = P_c \tilde{U}_{m g}(t,0) P_c$ for $m \in \{f,g\}$. If $[\tilde{U}_{rm}(T,t) a]P_c = 0$
for any $r,m \in \{g,f\}$, we can then swap $\tilde{U}_{mi}(T,t)$ and $a$ in Eq.~\eqref{eq:conditional_ET_eval} and obtain $\bra{r} \tilde{U}(T, t) \tilde{a}(t) U(t,0)\ket{g} P_c \propto a \tilde{U}_{rg}(T,0) P_c$. Such a condition is equivalent to
\begin{equation}
    [a, S(\vec{\phi})]P_c = [a, S^{\dagger}(\vec{\phi})]P_c = 0,
    \label{eq:SNAP_error_transparency}
\end{equation}
which is simply that the applied unitary $S(\vec{\phi})/S^{\dagger}(\vec{\phi})$ is error transparent to $a$. This can be satisfied by setting $S(\phi) = P_0 + P_3 + e^{i\theta}(P_2 + P_1)$.  

\subsection{$X$-axis rotation \label{appendix:ET_Z_rotation}}
Here, we show that the $X$-axis rotation in error-transparent to a single photon loss by showing the two involved SNAP gates (see Eq.~\eqref{eq:X-rotation}) satisfy the conditional error transparency in Def.~\ref{def:conditional_error_transparency}. Taking the first SNAP gate as an example, the proof is the same as that for the $Z$-axis rotation in the previous section, except that we now need to change $P_c$ to $D(\alpha) P_c$ when verifying Eq.~\eqref{eq:SNAP_error_transparency}. Recall that $S = e^{i\theta}P_{[s]} + I - P_{[s]}$ for the $X$-axis rotation, where $P_{[s]} = \sum_{i=0}^s\ket{i}\bra{i}$ is a projection into a neighborhood of vacuum. We take the large-$\alpha$ approximation $\ket{+_L} \approx \ket{C^+_{\alpha}}$ and $\ket{-_L} \approx \ket{C^+_{i\alpha}}$. Then,
\begin{equation}
\begin{aligned}
    a S D(\alpha) & \ket{+_L} \approx 2\alpha \ket{2\alpha} \approx  S a D(\alpha) \ket{+_L}, \\
    a S D(\alpha) & \ket{-_L} \approx a D(\alpha) \ket{-_L} \approx S a D(\alpha) \ket{-_L},
\end{aligned}
\label{eq:ET_X_rotation_eq}
\end{equation}
where we have used $S\ket{\beta} \approx \ket{\beta}$ for $|\beta| \gg 1$. Eq.~\eqref{eq:ET_X_rotation_eq} thus verifies that $S$ commutes with $a$ when acting on $D(\alpha) P_c$. 

\subsection{$XX$ rotation \label{appendix:ET_XX_rotation}}

Here, we show that the $XX$ rotation gate in Sec.~\ref{sec:XX_rotation} is robust against a small input phase rotation or the ancilla relaxation/dephasing or cavity loss that occurs in the middle. 

We first consider the input phase rotation error $E = e^{i (\delta\theta_a a^\dag a + \delta\theta_b b^\dag b)}$. Our aim is to show that
\begin{equation} \label{eq:UXXErrorClosure}
    U_{XX} E U_{XX}^\dag (P_c \otimes P_c) \approx E (P_c \otimes P_c).
\end{equation}

If we consider the SNAP gate $S$ to be replaced to the robust version $S = e^{-i\theta} P_{[s]} + (I - P_{[s]})$, we can write $U_{XX}$ in the following form,
\begin{equation} \label{eq:UXXform2}
    U_{XX} = e^{i\theta} P_{\langle[s],I\rangle} + (I - P_{\langle[s],I\rangle}),
\end{equation}
where
\begin{equation}
    P_{\langle[s],I\rangle} = BS(\frac{\pi}{2})^\dag \left( P_{[s]}\otimes I + I \otimes P_{[s]} \right) BS(\frac{\pi}{2}),
\end{equation}
is a projector onto the space where the input bosonic modes $A$ and $B$ have almost clean interference results (i.e., one output mode is close to vacuum) under the balanced beam-splitter $BS(\frac{\pi}{2})$. We can get Eq.~\eqref{eq:UXXform} in Sec.~\ref{sec:XX_rotation} from Eq.~\eqref{eq:UXXform2}. To prove Eq.~\eqref{eq:UXXErrorClosure}, we note that
\begin{equation} \label{eq:UXXECdeduce1}
\begin{aligned}
    &\quad U_{XX} E U_{XX}^\dag (P_c \otimes P_c) = E P_{\pm,\pm} + E P_{\pm,\mp} + \\
    & (e^{-i\theta}-1) (I - P_{\langle[s],I\rangle}) E P_{\pm,\pm} + (e^{i\theta}-1) P_{\langle[s],I\rangle} E P_{\pm,\mp},
\end{aligned}
\end{equation}
where 
\begin{equation}
\begin{aligned}
    P_{\pm,\pm} &= \ket{+_L,+_L}\bra{+_L,+_L} + \ket{-_L,-_L}\bra{-_L,-_L}, \\
    &\approx \ket{C_\alpha^+,C_\alpha^+}\bra{C_\alpha^+,C_\alpha^+} + \ket{C_{i\alpha}^+,C_{i\alpha}^+}\bra{C_{i\alpha}^+,C_{i\alpha}^+} \\
    P_{\pm,\mp} &= \ket{+_L,-_L}\bra{+_L,-_L} + \ket{-_L,+_L}\bra{-_L,+_L}, \\
    &\approx \ket{C_\alpha^+,C_{i\alpha}^+}\bra{C_\alpha^+,C_{i\alpha}^+} + \ket{C_{i\alpha}^+,C_\alpha^+}\bra{C_{i\alpha}^+,C_\alpha^+}.
\end{aligned}
\end{equation}
To simplify Eq.~\eqref{eq:UXXECdeduce1}, we notice that,
\begin{equation}
\begin{aligned}
    & \braket{C_{\alpha e^{i\delta\theta_a}}^+, C_{\alpha e^{i\delta\theta_b}}^+ | P_{\langle[s],I\rangle} | C_{\alpha e^{i\delta\theta_a}}^+, C_{\alpha e^{i\delta\theta_b}}^+ } \\
    =& \langle C_{\alpha e^{i\delta\theta_a}}^+, C_{\alpha e^{i\delta\theta_b}}^+ | BS(\frac{\pi}{2})^\dag ( P_{[s]}\otimes I + I \otimes P_{[s]} )\\ 
    & \quad  BS(\frac{\pi}{2})  | C_{\alpha e^{i\delta\theta_a}}^+, C_{\alpha e^{i\delta\theta_b}}^+ \rangle.
\end{aligned}
\end{equation}
Since
\begin{equation}
\begin{aligned}
     & BS(\frac{\pi}{2})  | C_{\alpha e^{i\delta\theta_a}}^+, C_{\alpha e^{i\delta\theta_b}}^+ \rangle \\
    =& \mu_\alpha^2 BS(\frac{\pi}{2}) \left( \ket{\alpha e^{i\delta\theta_a}} + \ket{-\alpha e^{i\delta\theta_a}}\right)_A \left( \ket{\alpha e^{i\delta\theta_b}} + \ket{-\alpha e^{i\delta\theta_b}}\right)_B \\
    =& \mu_\alpha^2 \bigg( \big( \ket{\frac{\alpha}{\sqrt{2}}(e^{i\delta\theta_a}+ e^{i\delta\theta_b})} \big)_A \big( \ket{\frac{\alpha}{\sqrt{2}}(e^{i\delta\theta_a}-e^{i\delta\theta_b})}  \big)_B \\
    &\quad + \ket{\frac{\alpha}{\sqrt{2}}(e^{i\delta\theta_a}- e^{i\delta\theta_b})} \big)_A \big( \ket{\frac{\alpha}{\sqrt{2}}(e^{i\delta\theta_a}+ e^{i\delta\theta_b})} \big)_B  \\
    &\quad + \ket{\frac{\alpha}{\sqrt{2}}(-e^{i\delta\theta_a}+ e^{i\delta\theta_b})} \big)_A \big( \ket{\frac{\alpha}{\sqrt{2}}(-e^{i\delta\theta_a}- e^{i\delta\theta_b})} \big)_B  \\
    &\quad + \ket{\frac{\alpha}{\sqrt{2}}(-e^{i\delta\theta_a}- e^{i\delta\theta_b})} \big)_A \big( \ket{\frac{\alpha}{\sqrt{2}}(-e^{i\delta\theta_a}+ e^{i\delta\theta_b})} \big)_B  \bigg).
\end{aligned}
\end{equation}
Here $\mu_\alpha = 1/\sqrt{2(1+\exp(-2|\alpha|^2))}$ is the normalization factor of the cat state.
when $|\delta\theta_a|,|\delta\theta_b|< \pi/8$, we have $|e^{i\delta\theta_a}-e^{i\delta\theta_b}|<2\sin(\pi/8)$. As a result, either the components on mode $A$ or the ones on mode $B$ will be almost covered in the region of $P_{[s]}$, which implies that when $\alpha \gg 1$, we can choose a value of $s=O(|\alpha|^2)$ such that
\begin{equation} \label{eq:XXrotCondition1}
    \braket{C_{\alpha e^{i\delta\theta_a}}^+, C_{\alpha e^{i\delta\theta_b}}^+ | P_{\langle[s],I\rangle} | C_{\alpha e^{i\delta\theta_a}}^+, C_{\alpha e^{i\delta\theta_b}}^+ } \to 1.
\end{equation}
Similarly, we will have
\begin{equation} \label{eq:XXrotCondition2}
\begin{aligned}
    \braket{C_{i\alpha e^{i\delta\theta_a}}^+, C_{i\alpha e^{i\delta\theta_b}}^+ | P_{\langle[s],I\rangle} | C_{i\alpha e^{i\delta\theta_a}}^+, C_{i\alpha e^{i\delta\theta_b}}^+ } \to 1, \\
    \braket{C_{\alpha e^{i\delta\theta_a}}^+, C_{i\alpha e^{i\delta\theta_b}}^+ | P_{\langle[s],I\rangle} | C_{\alpha e^{i\delta\theta_a}}^+, C_{i\alpha e^{i\delta\theta_b}}^+ } \to 0, \\
    \braket{C_{i\alpha e^{i\delta\theta_a}}^+, C_{\alpha e^{i\delta\theta_b}}^+ | P_{\langle[s],I\rangle} | C_{i\alpha e^{i\delta\theta_a}}^+, C_{\alpha e^{i\delta\theta_b}}^+ } \to 0.
\end{aligned}
\end{equation}
When Eq.~\eqref{eq:XXrotCondition1} and \eqref{eq:XXrotCondition2} hold, we can simplify Eq.~\eqref{eq:UXXECdeduce1} to
\begin{equation}
    U_{XX} E U_{XX}^\dag (P_c \otimes P_c) \approx E (P_{\pm,\pm} + P_{\pm,\mp}) = E (P_c\otimes P_c).
\end{equation}
i.e., $U_{XX}$ is robust against small phase rotation.

Now, we consider the error occurs during the process of $XX$ rotation.
First, we show that a single photon loss during the $XX$ rotation can only propagate to at most a single loss per mode by combing the idea of error transparency and error closure. Recall that a $XX$ rotation is implemented by two SNAP gates sandwiched by two BSs: $U_{XX} = \mr{BS}(S\otimes S)\mr{BS}$. We first show that a single photon loss during the BS can only propagate to an error of the form $c_1 a + c_2 b$, where $c_1, c_2 \in \mathbb{C}$, since the BS satisfies the error-closure condition. A BS is generated by $H_{\mr{BS}} = g_{\mr{BS}}a^{\dagger}b + g_{\mr{BS}}^{*} ab^{\dagger}$. Any single photon loss on one of the modes, e.g. $a$, satisfies the error closure condition since $\epsilon = \langle a, [H_{\mr{BS}}, a]\rangle = \langle a, b\rangle$, and for any $e = c_1 a + c_2 b$, where $c_1, c_2 \in \mathbb{C}$, we have: (i) $[H_{\mathrm{BS}},e] \in \epsilon$; (ii) $e$ is a correctable error if both $a$ and $b$ modes are encoded in a four-legged cat.

Next, we show that the SNAP gates $S\otimes S$ are error-transparent to both a single-photon loss during the gates and an input error of the form $c_1 a + c_2 b$. According to the analysis in Sec.~\ref{appendix:ET_Z_rotation}, the former error is satisfied if $[S\otimes S, a] \mr{BS}\,P_c = [S\otimes S, b] \mr{BS}\,P_c = 0$, which is a special case of the condition for the latter error:
\begin{equation}
    [S\otimes S, c_1 a + c_2 b] \mr{BS}\, P_c = 0.
    \label{eq:ET_condition_XX_SNAP}
\end{equation} 
To prove Eq.~\eqref{eq:ET_condition_XX_SNAP}, we take the approximation $\ket{+_L} \approx \ket{C^+_{\alpha}}$ and $\ket{-_L} \approx \ket{C^+_{i\alpha}}$ and show that $[S\otimes S, c_1 a + c_2 b] \mr{BS} \ket{\pm_L, \pm_L} = 0$. For $\ket{+_L, -_L}$,
\begin{equation}
\begin{aligned}
    & (S\otimes S) (c_1 a + c_2 b) \mr{BS} \ket{+_L, -_L} = (c_1 a + c_2 b) \mr{BS} \ket{+_L, -_L} \\
    =&  (c_1 a + c_2 b) (S\otimes S) \mr{BS} \ket{+_L, -_L},
\end{aligned}
\end{equation}
since $S\otimes S$ acts trivially on both $\mr{BS} \ket{+_L, -_L}$ and $(c_1 a + c_2 b)\mr{BS} \ket{+_L, -_L}$. The same argument also applies to $\ket{-_L, +_L}$. For $\ket{+_L, +_L}$, we have $\mr{BS}\ket{+_L, +_L} \propto (\ket{2\alpha, 0} + \ket{0,2\alpha} + \ket{-2\alpha, 0} + \ket{0, -2\alpha})$. Then
\begin{equation}
\begin{aligned}
    & \quad (S\otimes S) (c_1 a + c_2 b)\mr{BS} \ket{+_L, +_L} \\
    =& e^{i\theta}2\alpha [ c_1 (\ket{2\alpha, 0} + \ket{-2\alpha, 0}) + c_2(\ket{0,2\alpha} + \ket{0, -2\alpha})] \\
    =& (c_1 a + c_2 b) (S\otimes S)\mr{BS} \ket{+_L, +_L}.
\end{aligned}
\end{equation}
Similarly, we can show $[S\otimes S, c_1 a + c_2 b]\mr{BS} \ket{-_L, -_L} = 0$. 

Combining the error-closure property of the BSs and the error-transparency property of the SNAP gates, we conclude that a single photon loss during the execution of the $XX$ rotation can propagate to an error of the form $c_1^{\prime}a + c_2^{\prime} b$, which is correctable by the four-legged cats.

\section{More GPI examples \label{appendix:GPI_parity_flag}} 
Here, we provide more examples of ancilla-assisted bosonic operations that are GPI. 

\begin{figure}[h!]
    \centering
    \includegraphics[width=0.45\textwidth]{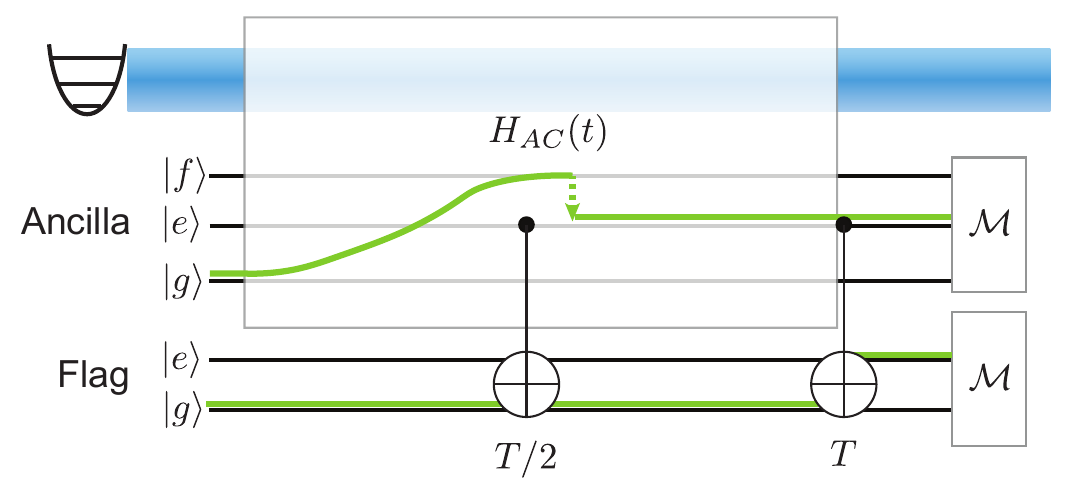}
    \caption{GPI SNAP gate with a flag qubit. The flag qubit is excited to $\ket{e}$ only if the ancilla decays from $\ket{f}$ to $\ket{e}$ at $t \in [T/2, T]$.}
    \label{fig:flagged_SNAP}
\end{figure}
Recall that the SNAP gate using a three-level transmon that we present in the main text (Sec.~\ref{sec:GPI_SNAP}) is GPI only if the $\chi$ mismatch $\Delta \chi$ is small than $\pi/2T$. In the scenario where $\Delta \chi \geq \pi/2T$, we can add another flag qubit~\cite{chao2018quantum, chao2020flag, Chamberland2018flag} to make the gate GPI.  

Notice that the major reason why the SNAP gate with a single ancilla is not GPI when $\Delta \chi$ is large is that the random dephasing range on the bosonic mode is too large due to the uncertainty of when an ancilla relaxation from $\ket{f}$ to $\ket{e}$ happens. Therefore, we can add an extra flag qubit to narrow down the ancilla-relaxation time window, and thus reducing the dephasing range. 

As shown in Fig.~\ref{fig:flagged_SNAP}, we apply two $X$ gates to the flag qubit controlled by the ancilla $\ket{e}$ state at time $T/2$ and $T$, respectively. As before, we consider adjacent-level relaxation errors for both the ancilla and the flag qubits, as well as arbitrary forms of dephasing errors. The flag qubit starts from $\ket{g}$ and only gets excited to $\ket{e}$ if the ancilla relaxes from $\ket{f}$ to $\ket{e}$ at a time $t \in [T/2, T]$. As such, a single ancilla relaxation incurs a random phase rotation of angle $\theta = \Delta \chi t$ on the bosonic mode, where $t \in  [0, T/2]$ if the flag qubit is measured in $\ket{g}$ while $t \in  (T/2, T]$ if the flag qubit is measured in $\ket{e}$. Formally, we can calculate the bosonic channels conditioned on the measurement outcomes of both the ancilla and flag qubits:



\begin{equation}
\begin{aligned}
\bbra{g,g} \mc{G}^{[1]}\kket{g,g}& \propto \mc{I},\\
 \bbra{f,g} \mc{G}^{[1]}\kket{g,g}& \propto S \bullet S^{\dagger},\\
   \bbra{e,g} \mc{G}^{[1]}\kket{g,g} & \propto \int_{\theta = 0}^{\Delta\chi T/2} S e^{-i\theta a^{\dagger}a} \bullet e^{i\theta a^{\dagger}a} S^{\dagger},\\
 \bbra{e,e} \mc{G}^{[1]}\kket{g,g} & \propto \int_{\theta = \Delta \chi T/2}^{\Delta \chi T}S e^{-i\theta a^{\dagger}a}\bullet e^{i\theta a^{\dagger}a} S^{\dagger},
\end{aligned}
\label{eq:flagged_SNAP_channels}
\end{equation}
where the first and second bits represent the ancilla and the flag qubit state for $\kket{\phi, \psi}$, respectively. 
According to Eq.~\eqref{eq:flagged_SNAP_channels}, the gate is $1$-GPI if $\Delta \chi T/2 < \pi/2$, or $\Delta \chi < \pi/T$. Therefore, we can allow twice as large $\chi$ mismatch by introducing another flag qubit. Note that we do not necessarily require the CNOT gates to be infinitely fast and noiseless, and Eq.~\eqref{eq:flagged_SNAP_channels} holds as long as the CNOT Hamiltonian is diagonal in the ancilla basis, e.g. $H_{\mr{CNOT}} \propto \ket{e}_A\bra{e}\otimes \left(\ket{e}_f\bra{g} + \ket{g}_f\bra{e}\right)$.

We remark that one can similarly construct a $1$-GPI parity measurement that can tolerate larger $\chi$ mismatch by introducing another flag qubit. 

\section{Details of numerical simulations}

Here, we provide the details of numerical simulations of the teleportation-based QEC and the parity-measurement shown in Fig.~\ref{fig:QEC_Infidelities}. In Fig.~\ref{fig:level_1_gadgets}(c) we have shown that the teleportation-based QEC circuit can be decomposed to a logical $\ket{+_L}$ state preparation gadget, a $X(\pi/2)$ gate on the input mode, two $Z(\pi/2)$ gates acting on two bosonic modes respectively, a $XX(\pi/2)$ gate on two bosonic modes, a logical $Z$ measurement on the input mode, and a potential $Z$ gate on the output mode. Note that the final $Z$ gate can be done in the software by updating the Pauli frame. 
The $X(\pi/2)$ and $XX(\pi/2)$ gates can further be decomposed to displacement operators, SNAP gates and/or beam-splitting operations following Fig.~\ref{fig:level_1_gadgets}(a) and (b).
The 1-FT $\ket{+_L}$-state preparation and the 1-FT logical $Z$ measurement are done by the procedures in Algorithm~\ref{alg:FT_X_preparation} and \ref{alg:FT_Z_measurement}, respectively, based on repeated single-shot parity/$Z$-basis measurement by dispersive-coupling to a $3$-level ancilla and majority vote.

In the numerical simulation, we assume the displacement operations can be performed quickly and ignore the faults that occur during them. We also assume a perfect preparation of the coherent state $\ket{\alpha}$ and the measurement on the $3$-level ancilla. On the other hand, we consider a noisy simulation of all other three basic gadgets including the dispersive coupling between a bosonic mode and a $3$-level ancilla, the SNAP gate, and the beam-splitting interaction. Below, we discuss the simulation details of these three noisy gadgets.

The dispersive coupling Hamiltonian $H_0$ is given by Eq.~\eqref{eq:SNAP_Hamiltonian1}. We set the dispersive coupling coefficient $\chi_f=2\pi\times 1 \mr{MHz}$. In the simulation, we mainly consider 4 types of Markovian noise: ancilla relaxation $J_{f\to e}= \sqrt{\gamma_{f\to e}}\ket{e}\bra{f}$, $J_{e\to g} = \sqrt{\gamma_{e\to g}} \ket{e}\bra{f}$, ancilla dephasing $J_{ph} = \sqrt{\gamma_{\phi}}(\ket{e}\bra{e} + 2\ket{f}\bra{f})$ and cavity loss $F_a = \sqrt{\kappa_1} a$.  

The SNAP gate Hamiltonian $\tilde{H}$ in the interaction picture of $H_0$ is given by Eq.~\eqref{eq:SNAP_interacting}. For the convenience of numerical simulation, we move to another interaction picture associated with a $\chi$-matched Hamiltonian,
\begin{equation}
    H_0' = - \chi_f(\ket{f}\bra{f} + \ket{e}\bra{e})\otimes a^{\dagger}a.
\end{equation}
In the interacting picture of $H_0'$, the SNAP gate Hamiltonian becomes
\begin{equation}
    \tilde{H}' = \Delta\chi \ket{e}\bra{e}\otimes a^\dag a + \Omega \left[|f\rangle \langle g|\otimes S(\vec{\phi}) + h.c.\right],
\end{equation}
where $\Delta\chi = \chi_f - \chi_e$. We set the Rabi drive strength $\Omega=0.3\chi_f$. 
The jump operators are then converted to $\tilde{J}_{f\to e}= \sqrt{\gamma_{f\to e}}\ket{e}\bra{f}$, $\tilde{J}_{e\to g} = \sqrt{\gamma_{e\to g}} \ket{e}\bra{f} e^{i\chi_f t a^\dag a}$, $\tilde{J}_{ph} = \sqrt{\gamma_{\phi}}(\ket{e}\bra{e} + 2\ket{f}\bra{f})$ and $\tilde{F}_a = \sqrt{\kappa_1} (P_g + e^{i\chi_f t}(P_e+P_f))\otimes a$ where $P_k:= \ket{k}\bra{k}$ for $k=g,e$, and $f$.
Note that $\tilde{J}_{e\to g}$ and $\tilde{F}_a$ are time-dependent which rotate quickly. To ease the simulation, we make a conservative estimate and approximate $\tilde{J}_{e\to g}$ by $\tilde{J}_{e\to g}' = \sqrt{\gamma_{e\to g}} \ket{e}\bra{f}\otimes e^{i\frac{\pi}{4} a^\dag a}$, i.e., as long as the $e\to g$ relaxation happens, a large dephasing error will be introduced on the cavity. To simplify $\tilde{F}_a$, we first notice that $\mc{D}[\tilde{F}_a] \approx \mc{D}[\sqrt{\kappa_1}P_g\otimes a] + \mc{D}[\sqrt{\kappa_1}(P_e+P_f)\otimes a]$ where we ignore all the fast rotating terms. This can be understood as a dephasing error between $P_g$ and $P_e + P_f$ introduced by the cavity loss. As a simple approximation, we merge the cavity-induced ancilla dephasing with the real ancilla dephasing, i.e., we set the cavity loss jump operator to be $\tilde{F}_a' = \sqrt{\kappa_1} a$, while the effective ancilla dephasing rate for $\tilde{J}_{f\to e}$ becomes $\gamma_\phi' = \gamma_\phi + \kappa_1/4$. The factor $1/4$ is introduced because we set $\Delta_f =2$ for $f$-level in $J_{ph}$.

The beam-splitting Hamiltonian is $H_{\mr{BS}} = ig_{\mr{BS}}(a b^\dag - a^\dag b)$ where $g_{\mr{BS}}$ is the BS interaction strength. We set the beam-splitting interaction strength $g_{\mr{BS}} = 2\chi_f$. The major noise we consider during the procedure are the cavity loss $F_a = \sqrt{\kappa_1} a$ and $F_b = \sqrt{\kappa_1} b$.

To simulate the dissipative time-evolution described above, we use the Monte-Carlo solver in the QuTiP package~\cite{johansson2012qutip}.
This can be easily done for a composite system with one bosonic mode and a 3-level ancilla. 
However, in the simulation of $XX(\pi/2)$ gate, we need to finish the following 3-step simulation:
\begin{enumerate}
    \item For a product input state $\ket{\psi}_A\otimes \ket{\psi'}_B$ on two bosonic modes $A$ and $B$, simulate the noisy beam-splitting interaction $\mathrm{BS}(\frac{\pi}{2})$. The output state is a two-mode entangled state $\ket{\Psi}_{AB}$.
    \item For the entangled state $\ket{\Psi}_{AB}$ input, simulate the noisy tensor-ed SNAP gate $S\otimes S$ with two $3$-level ancilla $A'$ and $B'$. The output state is a two-mode entangled state $\ket{\Psi'}_{AB}$.
    \item For the entangled state $\ket{\Psi'}_{AB}$ input, simulate the noisy beam-splitting interaction $\mathrm{BS}^\dag(\frac{\pi}{2})$. The output state is a two-mode entangled state $\ket{\Psi''}_{AB}$.
\end{enumerate}

The major bottleneck is the Step 2, where we need to consider a simulation of two bosonic modes $A$ and $B$ and two 3-level ancilla $A'$ and $B'$. To get rid of this costly simulation, we first perform a Schmidt decomposition on the entangled state $\ket{\Psi'}_{AB} = \sum_{k} \sqrt{p_k} \ket{u_k}_A \otimes \ket{v_k}_B$. 
Then, we simulate the SNAP gate on the bosonic mode $A$ and $B$ separately. 
Then, we simulate the SNAP gate on each components $\ket{u_k}$ or $\ket{v_k}$ separately, i.e. $(S\otimes S)\ket{\Psi^{\prime}}_{AB} = \sum_k \sqrt{p_k} (S\ket{u_k}_A)\otimes (S\ket{v_k}_B)$. Then, taking the simulation of the SNAP gate on mode $A$ as an example, we need to estimate when the quantum jumps occur and which jump operator occurs~\cite{johansson2012qutip}. This is determined by the following unnormalized expectation values
\begin{equation}
\begin{aligned}
    O_j(t) = \braket{\tilde{\Psi}(t)| J_j^\dag J_j|\tilde{\Psi}(t)}, 
\end{aligned}
\end{equation}
where
\begin{equation}
\begin{aligned}
    \ket{\tilde{\Psi}(t)} &= e^{-i t H_{\mr{eff}}} \ket{\Psi}_{AB}\otimes\ket{gg}_{A'B'}, \\
    H_{\mr{eff}} &= (H_0')_{AA'}\otimes I_{BB'} -\frac{i}{2} \sum_{j} J_j^\dag J_j.
\end{aligned}
\end{equation}
Here the summation of $j$ is taken for the 4 different jump operators we consider. We can simplify the form of $p_j$ to
\begin{equation}
\begin{aligned}
    & O_j(t) = \sum_k p_k O_j^{(k)}(t) \\
    &= \sum_{k} p_k \braket{u_k\otimes g | e^{it H_{\mr{eff}}} J_j^\dag J_j e^{-it H_{\mr{eff}}} |u_k\otimes g }.
\end{aligned}
\end{equation}
Here, $\ket{u_k\otimes g} := \ket{u_k}_A \otimes \ket{g}_{A'}$.
It is easy to verify that, for the ancillary relaxation and dephasing, the value of $O_j^{(k)}(t)$ for all the Schmidt components are the same. We also check numerically that for the cavity loss, the value of $O_j^{(k)}(t)$ for the Schmidt components with large weight $p_k$ are almost the same.
This means that we can approximate the overall simulation of $\ket{\Psi}_{AB}$ by fixing the quantum jump location of all the components $\{\ket{u_k}_A \otimes \ket{g}_{A'}\}$ to be the same. This can be done in the Monte-Carlo solver in QuTiP by passing the same random seed for all the Schmidt components.


%

\end{document}